\documentclass[amsmath,secnumarabic,floatfix,amssymb,nofootinbib,nobibnotes,letterpaper,11pt,tightenlines]{revtex4-2}

\usepackage{times}
\usepackage{latexsym,amsthm}
\usepackage{graphicx}
\usepackage[percent]{overpic}
\usepackage{units}
\usepackage{hyperref}
\usepackage{tikz-cd}
\usepackage{enumitem}
\setlist[enumerate]{topsep=0pt,itemsep=0ex,partopsep=1ex,parsep=1ex}
\usepackage{MnSymbol}

\usepackage{clrscode}

\newtheorem{theorem}{Theorem}

\newtheorem{lemma}[theorem]{Lemma}
\newtheorem{proposition}[theorem]{Proposition}
\newtheorem{corollary}[theorem]{Corollary}

\theoremstyle{definition}
\newtheorem{definition}[theorem]{Definition}

\newtheorem{claim}[theorem]{Claim}
\newtheorem{notation}[theorem]{Notation}
\newtheorem{algorithm}[theorem]{Algorithm}

\def\defn#1{Definition~\ref{def:#1}}
\def\eqn#1{\eqref{eq:#1}}
\def\thm#1{Theorem~\ref{thm:#1}}
\def\lem#1{Lemma~\ref{lem:#1}}
\def\figr#1{Figure~\ref{fig:#1}}
\def\prop#1{Proposition~\ref{prop:#1}}
\def\cor#1{Corollary~\ref{cor:#1}}

\def\clm#1{Claim~\ref{claim:#1}}

\newcommand{\abs}[1]{\lvert#1\rvert}
\newcommand{\R}{\mathbb{R}}

\newcommand{\Z}{\mathbb{Z}}

\newcommand{\norm}[1]{\left\| #1 \right\|}
\newcommand{\bdry}{\partial}
\newcommand{\bdy}{\bdry}

\newcommand{\gyradius}{\left< \operatorname{R^2_g} \right>}

\newcommand{\tr}{\operatorname{tr}}

\newcommand{\graphG}{\mathbf{G}}
\newcommand{\edgesE}{\mathbf{e}}
\newcommand{\verticesV}{\mathbf{v}}

\newcommand{\edge}{e}
\newcommand{\vertex}{v}

\newcommand{\head}{\operatorname{head}}
\newcommand{\tail}{\operatorname{tail}}

\newcommand{\im}{\operatorname{im}}

\newcommand{\proj}{\operatorname{proj}}

\newcommand{\VC}{\operatorname{VC}}
\newcommand{\VP}{\operatorname{VP}}
\newcommand{\EC}{\operatorname{EC}}
\newcommand{\ED}{\operatorname{ED}}
\newcommand{\Fr}{\operatorname{Fr}}
\newcommand{\ID}{\operatorname{ID}}

\newcommand{\dotp}[2]{\left< #1, #2 \right>}

\newcommand{\VCp}[2]{\dotp{#1}{#2}_{\tilde{L}^{-1}}}
\newcommand{\VPp}[2]{\dotp{#1}{#2}_{\tilde{L}^*}}
\newcommand{\ECp}[2]{\dotp{#1}{#2}}
\newcommand{\EDp}[2]{\dotp{#1}{#2}} 
\newcommand{\Vp}[2]{\dotp{#1}{#2}_{V}}
\newcommand{\Wp}[2]{\dotp{#1}{#2}_{W}}
\newcommand{\Up}[2]{\dotp{#1}{#2}_{U}}
\newcommand{\Frp}[2]{\dotp{#1}{#2}_{\Fr}}
\newcommand{\Vnorm}[1]{\norm{#1}_{V}}

\newcommand{\VCprod}{\left( \VC, \VCp{-}{-} \right)}
\newcommand{\VPprod}{\left( \VP, \VPp{-}{-} \right)}
\newcommand{\ECprod}{\left( \EC, \ECp{-}{-} \right)}
\newcommand{\EDprod}{\left( \ED, \EDp{-}{-} \right)}
\newcommand{\Vprod}{\left( V, \Vp{-}{-} \right)}
\newcommand{\Wprod}{\left( W, \Wp{-}{-} \right)}
\newcommand{\Uprod}{\left( U, \Up{-}{-} \right)}

\newcommand{\Hom}{\operatorname{Hom}}

\newcommand{\mat}[1]{#1}
\newcommand{\ones}[2]{\mat{\mathbf{1}}_{#1 \times #2}}
\newcommand{\onesVO}{\ones{\verticesV}{1}}
\newcommand{\onesOV}{\ones{1}{\verticesV}}
\newcommand{\onesVV}{\ones{\verticesV}{\verticesV}}

\newcommand{\onesEE}{\ones{\edgesE}{\edgesE}}
\newcommand{\onesDO}{\ones{d}{1}}

\newcommand{\onesON}{\ones{1}{n}}

\newcommand{\cyl}{\operatorname{cyl}}

\newcommand{\mean}[2]{\operatorname{\mathcal{E}}_{#1}\! \left(#2\right)}
\newcommand{\cov}[3]{\operatorname{cov}_{#1}\left( #2, #3 \right)}

\newcommand{\haus}[1]{\operatorname{\mathcal{H}}^{#1}}

\newcommand{\restricted}{\lefthalfcup}

\newcommand{\std}{\operatorname{std}}

\newcommand{\allverts}{\vertex_1, \dotsc, \vertex_\verticesV}
\newcommand{\alledges}{\edge_1, \dotsc, \edge_\edgesE}
\newcommand{\allvertstars}{X_{11}, \dotsc, X_{d\verticesV}}
\newcommand{\alledgestars}{W_{11}, \dotsc, W_{d\edgesE}}

\newcommand{\headtail}{\operatorname{ht}}

\newcommand{\mup}{\mu'}
\newcommand{\mug}{\mu_{\graphG}}
\newcommand{\mugw}{\mu^W_{\graphG}}
\newcommand{\mugqw}{\mu^{Q(W)}_{\graphG}}

\newcommand{\mugz}{\mu^0_{\graphG}}
\newcommand{\nug}{\nu_{\graphG}}

\newcommand{\mugwp}{(\mu')_{\graphG'}^{W'}}
\newcommand{\mugp}{(\mu')_{\graphG'}}

\newcommand{\nugp}{\nu'_{\graphG'}}

\newcommand{\mun}{\mu_n}
\newcommand{\mugn}{(\mun)_{\graphG_n}}
\newcommand{\mugwn}{\mugn^{W_n}}

\newcommand{\pushmu}{(f_1^*)_\sharp \mu}
\newcommand{\pushmugw}{(f_1^*)_\sharp \mugw}
\newcommand{\pushmug}{(f_1^*)_\sharp \mug}

\newcommand{\sinc}{\operatorname{sinc}}

\newcommand{\ftc}{\operatorname{ftc}}

\def\co{\colon}

\let\mgp=\marginpar \marginparwidth18mm \marginparsep1mm
\def\marginpar#1{\mgp{\raggedright\tiny #1}}

\let\lbl=\label
\def\label#1{\lbl{#1}\ifinner\else\marginpar{\ref{#1} #1}\ignorespaces\fi}

\begin{document}
\title[]{Random graph embeddings with general edge potentials}
\author{Jason Cantarella}
\altaffiliation{Mathematics Department, University of Georgia, Athens GA}
\noaffiliation
\author{Tetsuo Deguchi}
\altaffiliation{Ochanomizu University, Tokyo, Japan}
\noaffiliation
\author{Clayton Shonkwiler}
\altaffiliation{Department of Mathematics, Colorado State University, Fort Collins CO}
\noaffiliation
\author{Erica Uehara}
\altaffiliation{Ochanomizu University, Tokyo, Japan}
\noaffiliation


\begin{abstract}
In this paper, we study random embeddings of polymer networks distributed according to any potential energy which can be expressed in terms of distances between pairs of monomers. This includes freely jointed chains, steric effects, Lennard-Jones potentials, bending energies, and other physically realistic models.

A configuration of $n$ monomers in $\R^d$ can be written as a collection of $d$ coordinate vectors, each in $\R^n$. Our first main result is that entries from different coordinate vectors are uncorrelated, even when they are different coordinates of the same monomer. We predict that this property holds in realistic simulations and in actual polymer configurations (in the absence of an external field). 

Our second main contribution is a theorem explaining when and how a probability distribution on embeddings of a complicated graph may be pushed forward to a distribution on embeddings of a simpler graph to aid in computations. This construction is based on the idea of chain maps in homology theory. We use it to give a new formula for edge covariances in phantom network theory and to compute some expectations for a freely-jointed network. 
\end{abstract}
\date{\today}
\maketitle

\section{Introduction}
In the study of network polymers, it is common to represent the polymer topology by a graph and study the spatial distribution of the monomers in terms of the eigenvalues and eigenvectors of the Kirchhoff (or architecture~\cite{Kuchanov88}) matrix. In mathematics, this matrix is usually known as the graph Laplacian and studied as a discrete analogue of the usual Laplacian operator from continuum physics~\cite{Chung:1997tk}. The classical phantom network theory of James and Guth~\cite{James1947} restricts attention to the case where the probability distribution of positions of bonded pairs of monomers is described by a Gaussian spring and there are no other monomer-monomer interactions. Combining the linear algebra of the graph Laplacian with the simple behavior of Gaussian probability distributions under linear maps has allowed for exact calculations of striking simplicity and power (see~\cite{Yang1998,Wei1995b,Eichinger1985,tcrw,tcrw-theory}). 

However, mixing algebra and probability has a cost: presenting the theory this way makes it difficult to understand which results might generalize to different bond potentials. In this paper, we give a new formulation for the linear algebra of network polymers which is compatible with any combination of potentials between monomers which depend only on distance, such as Lennard-Jones, FENE, excluded-volume or elastic energy potentials, including fixed edge lengths, as in the case of the freely jointed chain. We note that different pairs of monomers are permitted to have different potentials, so we might model steric effects by a repulsive potential on some pairs and bonds by an attractive potential on others.

In the classical theory, a random embedding of graph $\graphG$ with $\verticesV$ vertices and $\edgesE$ edges into $\R^1$ is described by a vector space $\R^{\edgesE}$ of edge displacements and a vector space $\R^{\verticesV}$ of vertex positions connected by an incidence (or boundary) matrix $\bdy \co \R^{\edgesE} \rightarrow \R^{\verticesV}$.\footnote{The incidence matrix is usually called $B$ in the literature, but we use $\bdy$ here to match the notation used in the rest of the paper.} The Kirchhoff matrix $L$ is the $\verticesV \times \verticesV$ matrix $\bdy\bdy^T$. The special symmetries of Gaussian potentials allow us to study the problem of randomly embedding $\graphG$ into $\R^d$ coordinate-by-coordinate as a collection of $d$ independent one-dimensional problems. For our more general potentials, this will not be possible. For instance, if there is a fixed bond length between two monomers in space, the $x$, $y$, and $z$ coordinates of the vector between them are clearly not independent random variables. 

Our first idea, described in Section~\ref{sec:definitions}, is to replace the two vector spaces $\R^{\verticesV}$ and $\R^{\edgesE}$ with four vector spaces: $\EC \simeq \R^{\edgesE}$ and $\VC \simeq \R^{\verticesV}$, which are spaces of (scalar) weights on edges and vertices, and $\ED \simeq \R^{d\edgesE}$ and $\VP \simeq \R^{d\verticesV}$ which are spaces of vector edge displacements and vertex positions. These pairs of spaces are related by the contravariant functor $\Hom(-,\R^d)$, which exchanges the boundary map $\bdy \co \EC \rightarrow \VC$ for the displacement map $\bdy^* \co \VP \rightarrow \ED$. We then introduce natural inner products on these vector spaces which make $\bdy$ and $\bdy^*$ partial isometries (Propositions~\ref{prop:bdy is a partial isometry} and \ref{prop:bdystar is a partial isometry}).

Polymer models are specified in terms of a probability distribution on $\ED$. However, if the topology of the network $\graphG$ is nontrivial, only a subspace of configurations in $\ED$ correspond to valid embeddings of the graph. In this case, we have to condition our distribution on membership in the appropriate subspace. In James--Guth theory (where the probability distribution on $\ED$ is Gaussian), this presents no problems. However, in general there are some technical difficulties involved in conditioning an arbitrary probability distribution on a hypothesis of measure zero. In Section~\ref{sec:probability measures} we resolve this problem, giving in Definition~\ref{def:G-compatible}, Proposition~\ref{prop:disintegration with density}, and Corollary~\ref{cor:compatibility for joint distributions} easy-to-check conditions under which the construction can be mathematically justified.

In Section~\ref{sec:means and variances}, we consider the mean and variance of edge displacements and vertex positions, showing in~Propositions~\ref{prop:covariance structure for VP} and~\ref{prop:covariance structure for ED} that (as in James--Guth theory) different coordinates of these displacements or positions are uncorrelated,\footnote{Even though, as we noted above, they are~\emph{not} independent.} giving the covariance matrices a special structure. We are then able to give a general formula for the radius of gyration (Theorem~\ref{thm:gyradius formula}), which we use to compute the expected radius of gyration of a ring polymer whose edges each have an~\emph{arbitrary} symmetric probability distribution on $\R^d$ in terms of the edge variance (Proposition~\ref{prop:cycle graph gyradius}). It is then easy to recover the standard formulae for the expected radius of gyration of the freely jointed ring (Corollary~\ref{cor:freely jointed ring}) and the Gaussian ring polymer (Corollary~\ref{cor:Gaussian ring}).

It's common in phantom network theory to try to relate the distribution of embeddings of a complicated graph to the distribution of embeddings of a simpler graph obtained, for instance, by contracting the graph by deleting ``bifunctional\footnote{That is, degree 2.}'' vertices. There are many variations of this construction. We unify the theory of such strategies in Section~\ref{sec:chain maps} by borrowing the idea of~\emph{chain maps} from homology theory. Our main result, Theorem~\ref{thm:chain maps and probability}, gives precise information on when and how a probability measure may be pushed from a more a complicated graph to a simpler one. In our final results, we demonstrate the utility of this proposition by giving a new, simple method for computing edge covariances in James--Guth theory~(\prop{projections}) and numerically computing the expectation of junction-junction distance in a tetrahedral network whose edges are freely-jointed chains (Figure~\ref{fig:numerical integration versus markov}).

Various useful results from linear algebra are reviewed in Appendix~\ref{sec:background} and referenced throughout. Everything in the appendix is basically standard, but we are particularly interested in using nonstandard inner products (so that adjoints and transposes don't coincide), as well as more general spaces of linear transformations than just dual spaces, so none of this material is presented in quite the way we need in any textbooks we are familiar with. 

\section{Definitions}
\label{sec:definitions}

We begin with some definitions.

\begin{definition}\label{def:chainspaces}
Given a multigraph $\graphG$ with vertices $\allverts$, we define the space $\VC$ of \emph{vertex chains}\footnote{The terminology comes from homology theory~\cite{Hatcher:2002ut}, which will be a continuing inspiration for our point of view.} to be the vector space of formal linear combinations $x = x_1 \vertex_1 + \dots + x_\verticesV \vertex_\verticesV$. The vertices form a canonical basis for this space, which is isomorphic to $\R^\verticesV$. 

If $\graphG$ has edges $\alledges$, we define the space $\EC$ of \emph{edge chains} to be the vector space of formal linear combinations $w = w_1 \edge_1 + \dots + w_\edgesE \edge_\edgesE$. The edges form a canonical basis for this space, which is isomorphic to $\R^\edgesE$. 
\end{definition}

The vector spaces $\VC$ and $\EC$ are joined by a natural linear map:
\begin{definition}\label{def:boundary}
The \emph{boundary map} $\bdy \co \EC \rightarrow \VC$ is defined by $\bdy (\edge_i) = \head(\edge_i) - \tail(\edge_i)$.
\end{definition}
The name ``boundary map'' comes from a natural convention: the (signed) boundary of a oriented edge consists of two ``oriented'' vertices: the head vertex with orientation $+1$ and the tail vertex with orientation $-1$. The transpose $\mat{\bdy}^T$ is often called the~\emph{incidence matrix} of $\graphG$, although we are choosing a particular convention for how loop edges are recorded in the incidence matrix. 

Vertex and edge chains are linear combinations of vertices and edges; we think of them as weights on the vertices and edges. Varying the weights allows us to pick out various subsets and averages of vertices in the graph. For instance,  $1 \vertex_i$ represents vertex $i$ alone, while $\frac{1}{\verticesV} (\vertex_1 + \cdots + \vertex_\verticesV) = \frac{1}{\verticesV} \onesVO$ represents a sort of average vertex,\footnote{We use the notation $\mat{\mathbf{1}}_{n \times m}$ for the $n \times m$ matrix containing all $1$s.} in the sense that evaluating any linear functional on this chain gives the average of the functional over the whole graph. 

\begin{definition}\label{def:loop space}
The subspace $\ker \bdy \subset \EC$ is called the \emph{loop space} of $\graphG$.
\end{definition}

The name comes from the fact that if the oriented edges $\edge_1, \dotsc, \edge_n$ form a closed loop, then there are vertices $\vertex_1, \dotsc, \vertex_n$ so that $\bdy \edge_i = \vertex_{i+1} - \vertex_i$ for $i \in 1, \dotsc, n-1$ and $\bdy \edge_n = \vertex_1 - \vertex_n$. Thus $\bdy (\edge_1 + \cdots + \edge_n) = 0$. 

\begin{proposition}\label{prop:loop space}
Every $w \in \ker \bdy \subset \EC$ is a linear combination of closed loops. The dimension of $\ker \bdy$ is the cycle rank $\xi(\graphG) = \edgesE - \verticesV + 1$ of $\graphG$.
\end{proposition}

For a proof of the above proposition in the context of a gentle introduction to homology on graphs, see Chapter~4 in Sunada's book~\cite{Sunada:2013jt}.

We now want to consider assignments of vectors (instead of scalar weights) to our vertices and edges.

\begin{definition}\label{def:VP}
Given a graph $\graphG$ with vertices $\allverts$, the vector space $\Hom(\VC,\R^d)$ of linear maps $X \co \VC \rightarrow \R^d$ is the \emph{vertex positions space}, denoted $\VP$.  
\end{definition}

Each $X \in \VP$ describes an embedding of $\graphG$ in $\R^d$: if we think of $X \in \VP$ as a $d \times \verticesV$ matrix, then the $(i,j)$ entry is the $i$th coordinate of the position of vertex $\vertex_j$. This is the same as the standard basis for $\Hom(\VC,\R^d)$ that appears in~\defn{hom}, which we now call $X_{ij}$: $X_{ij}(x_1 \vertex_1 + \cdots + x_\verticesV \vertex_\verticesV) = (0, \dotsc, x_j, \dotsc, 0)$, where $x_j$ is in the $i$th position.

As a linear map, $X$ takes a weighted sum of (abstract) vertices to the corresponding weighted sum of their positions in $\R^d$. For instance, $\frac{1}{\verticesV} X(v_1 + \dots + v_\verticesV) = \frac{1}{\verticesV} \mat{X} \onesVO$ is the position of the center of mass of the vertices.

\begin{definition}\label{def:ED}
If $\graphG$ has edges $\alledges$, the vector space $\Hom(\EC,\R^d)$ of linear maps $W \co \EC \rightarrow \R^d$ is the \emph{edge displacements space}, denoted $\ED$. 
\end{definition}

This space associates a vector in $\R^d$ with every edge of $\graphG$, rather than every vertex.
If we represent $W \in \ED$ by a $d \times \edgesE$ matrix, then the $(i,j)$ entry is the $i$th coordinate of the vector associated with $\edge_j$. Again, this is the same as the standard basis for $\Hom(\EC,\R^d)$ that appears in~\defn{hom}: $W_{ij}(w_1 \edge_1 + \cdots + w_\edgesE \edge_\edgesE) = (0, \dotsc, w_j, \dotsc, 0)$, where $w_j$ is in the $i$th position.

As a linear map, $W$ maps a weighted sum of (abstract) edges to the corresponding weighted sum of vectors associated with those edges. 

\begin{definition}\label{def:displacement map}
The displacement map $\operatorname{disp}: \VP \rightarrow \ED$ is defined by 
\begin{equation*}
\operatorname{disp}(X)(\edge_i) := X(\head \edge_i) - X(\tail \edge_i).
\end{equation*}
\end{definition}

Every embedding of the vertices of $\graphG$ in $\R^d$ given by an $X \in \VP$ has a corresponding set of displacement vectors $W = \operatorname{disp}(X) \in \ED$. However, not every $W \in \ED$ is derived from a set of positions for the vertices. For instance, if $\graphG$ is the cycle graph with three edges $\edge_1 = \vertex_1 \rightarrow \vertex_2$, $\edge_2 = \vertex_2 \rightarrow \vertex_3$, and $\edge_3 = \vertex_3 \rightarrow \vertex_1$, any $W = \operatorname{disp}(X)$ must have $W(\edge_1 + \edge_2 + \edge_3) = \vec{0} \in \R^d$.

We now give a surprising connection between $\operatorname{disp}$ and our boundary map $\bdy \co \EC \rightarrow \VC$. Recall that any linear map $F \co V \rightarrow W$ induces a corresponding linear map $F^* \co \Hom(V,U) \rightarrow \Hom(W,U)$ as defined in~\eqref{eq:induced hom map}. 

\begin{definition}\label{def:bdystar}
The map $\bdy^* \co \VP \rightarrow \ED$ is the linear map induced by $\bdy \co \EC \rightarrow \VC$.
\end{definition}

\begin{proposition}\label{prop:disp is coboundary}
We have $\operatorname{disp} = \bdy^*$. 
\end{proposition}

\begin{proof}
We observe that in our bases, $\bdy^*(X_{ik}) = \sum_{j=1}^\edgesE \headtail_{kj} W_{ij}$, where 
\begin{equation*}
\headtail_{kj} = 
\begin{cases}
0, & \text{if $\vertex_k = \head(e_j) = \tail(e_j)$}\\
+1, & \text{if $\vertex_k = \head(e_j)$} \\
-1, & \text{if $\vertex_k = \tail(e_j)$} \\
0, & \text{otherwise}
\end{cases}
\end{equation*}
It follows that $\bdy^*(X)(\edge_j) = X(\head(\edge_j) - \tail(\edge_j)) = \operatorname{disp}(X)(\edge_j)$.
\end{proof}
It follows immediately that
\begin{lemma}\label{lem:net displacement}
We say $P \in \EC$ is a path from $\vertex_i$ to $\vertex_j$ if $\bdy P = \vertex_j - \vertex_i$. Then $(\bdy^* X)(P) = (\operatorname{disp} X)(P) = X(\vertex_j) - X(\vertex_i)$ is the net displacement between the ends of the path $P$.
\end{lemma}

We now want to characterize $\ker \bdy^*$ and $\im \bdy^*$. 

\begin{proposition}
\label{prop:ker and im of bdystar}
If $\graphG$ is a connected graph, then 
\begin{align*}
\ker \bdy^* &= \{ X \in \VP : X(\vertex_i) = X(\vertex_j) \in \R^d \text{ for all } i,j \in 1, \dotsc, \verticesV \} \\
&= \{ Z \otimes \onesVO : \text{$Z \in \R^d$ is a $d \times 1$ column vector} \} \\
\im \bdy^* &= \{ W \in \ED : W(u) = 0 \text{ for all $u \in \ker \bdy \subset \EC$} \}.
\end{align*}
As a consequence, $\dim \ker \bdy^* = d$ and $\dim \im \bdy^* = d(\verticesV-1)$.
\end{proposition}

\begin{proof}
Using \prop{annihilator props}, $\ker \bdy^*$ is the annihilator $(\im \bdy)^0$ of $\im \bdy$. In other words, if $X \in \ker \bdy^*$, then $0 = X(\bdy \edge_i) = X(\head \edge_i) - X(\tail \edge_i)$, so $X(\head \edge_i) = X(\tail \edge_i)$. Since $\graphG$ is connected, this implies that $X(\vertex_i)$ is the same for all vertices $\vertex_i$. Similarly, $\im \bdy^* = (\ker \bdy)^0$, which completes the proof.
\end{proof}

This proposition means that two configurations of vertices $X$ and $Y$ in $\VP$ have the same edge displacements $\bdy^* X = \bdy^* Y$ if and only if they are translations of each other. In our description of $\ker \bdy^*$ as the set of maps in the form $Z \otimes \onesVO$, the vector $Z \in \R^d$ is the translation vector.

We have now reached an important point: the space $\ED$ is the space of \emph{arbitrary} assignments of vectors $W(\edge_i) \in \R^d$ to the edges of $\graphG$. However, only $W \in \im \bdy^* \subset \ED$ are assignments of vectors which are displacements between a choice of vertex positions $X \in \VP$. \prop{ker and im of bdystar} tells us that $W \in \im \bdy^*$ if and only if the total displacement  around any $v$ in the loop space of $\graphG$ is zero.

Casassa~\cite{Casassa1965} adopted this point of view for ring polymers, where the loop space is one-dimensional and spanned by $\edge_1 + \cdots + \edge_\edgesE$,\footnote{At least, this is true if we orient the edges consistently around the loop. If not, we'd have to reverse some signs in the sum to arrive at a consistent orientation for the ring.} by observing that a set of edge displacements form a closed ring if and only if their sum is the zero vector.

However, he did not generalize this point of view to other network topologies: for a more complicated graph, it's clear that there are infinitely many possible loops, but without characterizing the loops as $\ker \bdy$, it's not at all clear that the loops form a subspace and hence that it suffices to require that total displacements around a finite basis for $\ker \bdy$ vanish.

\subsection{The graph Laplacian}

We now introduce the graph Laplacian,\footnote{The graph Laplacian is also known as the Kirchhoff adjacency matrix. It is central to the theory of James and Guth~\cite{James1947}, and also to Flory~\cite{Flory1976} and Eichinger~\cite{Eichinger1972} as a quadratic form expressing the potential energy of a phantom network with Gaussian chains joining the junctions. We will put it to more general use.} which will be a key part of the story.

\begin{definition}\label{def:L}
The \emph{graph Laplacian} $L \co \VC \rightarrow \VC$ is defined by $L = \bdy \bdy^T$. Thought of as a matrix with respect to the standard basis $\allverts$ for $\VC$, we have
\begin{equation}
\mat{L}_{ij} = 
\begin{cases}
\deg(\vertex_i) - 2 \# \text{(loop edges $\vertex_i \rightarrow \vertex_i$)}, & \text{if $i = j$,} \\
-\# \text{(edges $\vertex_j \rightarrow \vertex_i$)} - \# \text{(edges $\vertex_i \rightarrow \vertex_j$)}, 
& \text{if $i \neq j$}.
\end{cases}
\end{equation}
\end{definition}

Much is known~\cite{Chung:1997tk} about the graph Laplacian and how it reveals various properties of the (multi)graph~$\graphG$. We will record a couple of useful facts here.

\begin{proposition}\label{prop:basic Laplacian properties}
The $\verticesV \times \verticesV$ matrix $\mat{L}$ is symmetric and positive semidefinite. We have
\begin{equation*}
\im L = \im \bdy = \ker \onesVV \quad\text{and}\quad \ker L = \ker \bdy^T = \im \onesVV.
\end{equation*}
\end{proposition}

\begin{proof} 
Since $\bdy^T = \bdy^*$ if $d=1$,~\prop{ker and im of bdystar} tells us that $\ker \bdy^T$ is spanned by the constant chains $\onesVO \in \VC$, which are the image of $\onesVV = \onesVO \onesOV$. Since $\ker \bdy^T \subset \ker L$, this tells us that $\im \onesVV \subset \ker L$. Similarly, if $\bdy^T \onesVO = 0$, then $\onesOV \bdy = 0$, or $\im L \subset \ker \onesVV$. 

But $\im \bdy^T = \star \im \bdy^* = \star (\ker \bdy)^0$ only intersects $\ker \bdy$ at the origin. Thus $\ker L = \ker \bdy^T$ is one-dimensional. It follows that $\im L$ is $\verticesV - 1$ dimensional. Since $\im \onesVV$ is one-dimensional, $\ker \onesVV$ is $\verticesV - 1$ dimensional. The result now follows from the inclusions above.
\end{proof}

Since $L$ has a kernel, we cannot invert it. However, we can get an invertible operator by defining an operator which is the identity on $\ker L$ rather than collapsing it:

\begin{definition}\label{def:Ltilde}
The \emph{augmented graph Laplacian} is the operator $\tilde{L} \co \VC \to \VC$ which is represented in the basis $\allverts$ by the $\verticesV \times \verticesV$ matrix $\mat{\tilde{L}} = \mat{L} + \frac{1}{\verticesV} \onesVV$. 
\end{definition}

\begin{proposition}\label{prop:im of subspaces under Ltilde}
$\ker \tilde{L} = \{0\}$ and $\im \tilde{L} = \VC$. Thus $\tilde{L}$ is invertible. Further,
\[
\tilde{L}(\im \onesVV) = \im \onesVV \quad\text{and}\quad \tilde{L}(\im \bdy) = \im \bdy.
\] 
\end{proposition}

\begin{proof}
Since $L$ is symmetric, it is self-adjoint in the standard $\dotp{-}{-}$ inner product on $\VC$ and, by \lem{orthogonal decomposition}, $\VC = \ker L \oplus \im L$ is an orthogonal decomposition of $\VC$. Using~\prop{basic Laplacian properties}, this means that any $x \in \VC$ can be written as $x = \onesVV y + \bdy z$. Further,
\begin{equation*}
\tilde{L}x = (L + \frac{1}{\verticesV} \onesVV)(\onesVV y + \bdy z) 
    	   = \onesVV y + L \bdy z.
\end{equation*}
Since $L \bdy z$ and $\onesVV y$ are in the orthogonal subspaces $\im L$ and $\ker L$, we can first conclude that $\tilde{L}x = 0$ if and only if $x=0$. Further, if $x \in \im \onesVV$, then $z = 0$, and $\tilde{L}x = \onesVV y \in \im \onesVV$ while if $x \in \im \bdy$, then $y = 0$ and $\tilde{L}x = L \bdy z \in \im L = \im \bdy$. Thus $\tilde{L}(\im \onesVV) \subset \im \onesVV$ and $\tilde{L}(\im \bdy) \subset \im \bdy$. Counting dimensions yields the reverse inclusions immediately.
\end{proof}

\subsection{Inner products}

Up until this point, we have had vector spaces and linear maps, but (except as a convenience in the proof of \prop{im of subspaces under Ltilde}) not inner product spaces. Our next goal is to introduce natural inner products on all four spaces $\VC$, $\EC$, $\VP$, and $\ED$. While the inner products on $\EC$ and $\ED$ will be the expected ones, the inner products on $\VC$ and $\VP$ will be non-standard.

\begin{definition}
The inner product space $\VCprod$ is the vector space $\VC$, together with the inner product given in the $\allverts$ basis by $\tilde{L}^{-1}$. The inner product space $\ECprod$ is the vector space $\EC$, together with the standard inner product in the $\alledges$ basis.
\label{def:VCprod and ECprod}
\end{definition}

The induced inner products make $\VP$ and $\ED$ inner product spaces as well:
\begin{definition}
The inner product space $\VPprod$ is the vector space $\VP$, together with the inner product given in the $\allvertstars$ basis by $\mat{\tilde{L}^*} = \mat{I_d} \otimes \mat{\tilde{L}}$. The inner product space $\EDprod$ is the vector space $\ED$, together with the standard inner product in the $\alledgestars$ basis.
\label{def:VPprod and EDprod}
\end{definition}

Adjoints, orthogonality, Moore--Penrose pseudoinverses, and singular value decompositions all depend on inner products, so in principle we must be careful when using any of these (or referring to the literature) that our results hold in the desired inner product. This is simplified by the following useful fact:

\begin{proposition}
The operators $\bdy^+$, $\bdy^{T+}$ and $L^+ = \bdy^{T+} \bdy^+$ in the basis $\allverts$ are the same whether they are computed with respect to the $\VCp{-}{-}$ or the $\dotp{-}{-}$ inner product on $\VC$. 
\label{prop:samesies}
\end{proposition}

\begin{proof}
Since the first two Moore--Penrose properties $F^+ F F^+ = F^+$ and $F F^+ F = F$ don't refer to an inner product, pseudoinverses computed with respect to any inner product obey these conditions.

To see that $\bdy$ has the same Moore--Penrose pseudoinverse with respect to both the $\VCp{-}{-}$ and $\dotp{-}{-}$ inner products on $\VC$, it suffices to check that if  $\mat{\bdy^+ \bdy}$ and $\mat{\bdy \bdy^+}$ are symmetric (that is, they are self-adjoint in $\dotp{-}{-}$ on both $\VC$ and $\EC$), then $\bdy \bdy^+$ and $\bdy^+ \bdy$ are self-adjoint (in $\VCprod$ and $\ECprod$). Since the inner product on $\ECprod$ is standard, $(\bdy^+ \bdy)^\dag = (\bdy^+ \bdy)^T = \bdy^+ \bdy$ immediately. So suppose we have computed $\bdy^+$ with respect to $\dotp{-}{-}$ on both $\VC$ and $\EC$. 
It's known~(\cite[eq.\ 7]{Ghosh2008}) that if $L^+$ is computed with respect to $\dotp{-}{-}$ on $\VC$, then 
\begin{equation}
\tilde{L}^{-1} = L^+ + \frac{1}{\verticesV} \onesVV = \bdy^{T+} \bdy^+ + \frac{1}{\verticesV}\onesVV.
\label{eq:Ltilde inverse}
\end{equation}
Thus, using \lem{adjoint formula}, 
\begin{equation*}
(\bdy \bdy^+)^\dag = \tilde{L} (\bdy \bdy^+)^T \tilde{L}^{-1} = 
(L + \frac{1}{\verticesV} \onesVV) \bdy \bdy^+ (L^+ + \frac{1}{\verticesV} \onesVV). 
\end{equation*}
\prop{basic Laplacian properties} tells us that $\onesVV \bdy = 0$ since $\im \bdy = \ker \onesVV$. Thus, we can simplify the right hand side above and get
\begin{equation*}
(\bdy \bdy^+)^\dag = 
L (\bdy \bdy^+)(L^+ + \frac{1}{\verticesV} \onesVV) = 
L (\bdy^{T+} \bdy^T)(L^+ + \frac{1}{\verticesV} \onesVV),
\label{eq:self-adjoint 2}
\end{equation*}
where we used $\bdy \bdy^+ = (\bdy \bdy^+)^T = \bdy^{T+} \bdy^T$. Again,~\prop{basic Laplacian properties} tells us that $\bdy^T \onesVV = 0$ since $\im \onesVV = \ker \bdy^T$. So we can simplify the right-hand side above and get
\begin{equation}
(\bdy \bdy^+)^\dag = 
L(\bdy^{T+} \bdy^T)L^+ = 
\bdy (\bdy^T \bdy^{T+}) (\bdy^T \bdy^{T+}) \bdy^+ = 
\bdy \bdy^+ \bdy \bdy^+ = \bdy \bdy^+ ,
\end{equation}
where we used $ \bdy^{T} \bdy^{T+} = (\bdy^+ \bdy)^T = \bdy^+ \bdy$ and $(\bdy^+ \bdy)(\bdy^+ \bdy) = \bdy^+ \bdy$. This proves that $\bdy^+$ does not depend on whether we compute in $\VCprod$ or $(\VC,\dotp{-}{-})$.

To prove the second part, assume that $(\bdy^T)^+$ has been computed with respect to $\dotp{-}{-}$, so that $\bdy^{T+} \bdy^T$ and $\bdy^T \bdy^{T+}$ are symmetric. We must show that $\bdy^{T+} \bdy^T$ and $\bdy^T \bdy^{T+}$ are self-adjoint. But $\bdy^{T+} \bdy^T = (\bdy \bdy^+)^T = \bdy \bdy^+$, which we just proved is self-adjoint in $\VCprod$, and $\bdy^T \bdy^{T+} = (\bdy^+ \bdy)^T = \bdy^+ \bdy$, which is symmetric and hence self-adjoint in $\ECprod$.

Last, if we assume that $L^+$ has been computed with respect to $\dotp{-}{-}$, we know that $LL^+$ and $L^+L$ are symmetric and need to show they are self-adjoint. We note first that $L L^+ = \proj_{\im L}$ (in $\dotp{-}{-}$) and $L^+ L = \proj_{\im L^T} = \proj_{\im L}$ since $L = L^T$. Thus $L L^+ = L^+ L$.
Now this proof goes exactly along the lines of the first one. As above,
\begin{equation*}
(L^+L)^\dag = \tilde{L}(L^+ L)^T\tilde{L}^{-1} = (L + \frac{1}{\verticesV} \onesVV)(L L^{T+})\tilde{L}^{-1},
\end{equation*}
where we used $L^T = L$. But $\im L = \ker \onesVV$, so this simplifies to 
\begin{equation*}
(L^+ L)^\dag = L (L L^{T+}) \tilde{L} = L (L^+ L) (L^+ + \frac{1}{\verticesV} \onesVV),
\end{equation*}
where we've used $L L^{T+} = L^T L^{T+} = (L^+ L)^T = L^+ L$. Since $\im \onesVV = \ker L$, we have
\begin{equation*}
 (L^+ L)^\dag = L (L^+ L) (L^+ + \frac{1}{\verticesV} \onesVV) = L L^+ L L^+ = L L^+ = L^+ L,
 \end{equation*}
 which completes the proof.
\end{proof}

As an immediate consequence, $\im \bdy^+ = \im \bdy^\dag = \im \bdy^T$, and the orthogonal projections
$\bdy \bdy^+ = \proj_{\im \bdy}$, $\bdy^+ \bdy = \proj_{\im \bdy^+}$ are the same in either inner product. Similarly, $L L^+ = \proj_{\im L} = \proj_{\im \bdy}$ is the same in either inner product.

\begin{corollary}
$\bdy^{*+} \co \EDprod \rightarrow \VPprod$ and $\bdy^{*+} \co \EDprod \rightarrow (\VP, \dotp{-}{-})$ are the same operator. 
\label{cor:samesies star}
\end{corollary}

\begin{proof}
We know from \prop{starplus is plusstar} that in the inner products on $\Hom(\VC,\R^d)$ and $\Hom(\EC,\R^d)$ induced by inner products on $\VC$, $\EC$, and $\R^d$, $(\bdy^*)^+ = (\bdy^+)^* = \bdy^{*+}$. Since we just proved that $\bdy^+$ is the same operator in either inner product on $\VC$, this implies that $\bdy^{*+}$ is as well.
\end{proof} 

\subsection{Partial isometries}

We have now done some careful technical work to set things up, and can start to collect some of the rewards. We will see that with respect to our inner products, the maps $\bdy$ and $\bdy^*$ have extremely nice properties. We first recall a definition from functional analysis:

\begin{definition}
A map $A \co \Vprod \rightarrow \Wprod$ is a \emph{partial isometry} if, for all $x, y \in (\ker A)^\perp$, we have $\Vp{x}{y} = \Wp{Ax}{Ay}$. We call $(\ker A)^\perp = \im A^\dag = \im A^+$ the \emph{initial space} of the partial isometry and $\im A$ the \emph{final space} of the isometry.
\label{def:partial isometry}
\end{definition}

\begin{proposition}
The map $\bdy \co \ECprod \rightarrow \VCprod$ is a partial isometry.
\label{prop:bdy is a partial isometry}
\end{proposition}

\begin{proof}
The proof is a computation. Using~\prop{samesies} and~\eqn{Ltilde inverse} we have
\begin{align*}
\VCp{\bdy u}{\bdy w} = \dotp{\bdy u}{\tilde{L}^{-1} \bdy w} = \dotp{\bdy u}{L^+ \bdy w} + \frac{1}{\verticesV}\dotp{\bdy u}{\onesVV \bdy w}.
\end{align*}
On the far right $\dotp{\bdy u}{\onesVV \bdy w} = \dotp{u}{\bdy^T \onesVV \bdy w} = 0$,
since $\im \onesVV = \ker \bdy^T$ by~\prop{basic Laplacian properties}. Again using~\prop{samesies}, we are left with 
\begin{equation*}
\VCp{\bdy u}{\bdy w} = \dotp{\bdy u}{L^+ \bdy w} = \dotp{\bdy u}{\bdy^{T+} \bdy^+ \bdy w} = \dotp{\bdy^+ \bdy u}{\bdy^+ \bdy w} = \dotp{\proj_{\im \bdy^+} u}{\proj_{\im \bdy^+} w}.
\end{equation*}
Thus if $u, w$ are in the intial space $\im \bdy^+$, we have $\dotp{\proj_{\im \bdy^+} u}{\proj_{\im \bdy^+} w} = \dotp{u}{w}$. Note that $\dotp{-}{-}$ is our inner product on $\ECprod$, so we have completed the proof.
\end{proof}

\begin{proposition}
The map $\bdy^* \co \VPprod \rightarrow \EDprod$ is a partial isometry.
\label{prop:bdystar is a partial isometry}
\end{proposition}

\begin{proof}
	Suppose $X,Y \in \VP$. Now $\tilde{L} = \bdy\bdy^T + \frac{1}{\verticesV}\onesVV$, so $\tilde{L}^* = (\bdy^{T*})(\bdy^*) + \frac{1}{\verticesV} \onesVV^*$ and
	\begin{equation}\label{eq:Bstar partial isometry}
		\VPp{X}{Y} = \dotp{X}{\bdy^{*T} \bdy^* Y} + \dotp{X}{\frac{1}{\verticesV}\onesVV^* Y} = \dotp{\bdy^*X}{\bdy^*Y} + \dotp{X}{\frac{1}{\verticesV}\onesVV^* Y}.
	\end{equation}	
Suppose $X, Y$ are in the initial space $\im (\bdy^*)^\dag$. Since $\im (\bdy^*)^\dag = \im (\bdy^*)^T$,
we know $Y = (\bdy^{*T})W$ for some $W \in \ED$. Then the second term in~\eqref{eq:Bstar partial isometry} is equal to
	\begin{equation*}
		\dotp{X}{\frac{1}{\verticesV}\onesVV^* \bdy^{T*} W}  = \frac{1}{\verticesV} \EDp{X}{(\bdy^T \onesVV)^* W} = 0,
	\end{equation*}
	since $\im \onesVV =  \ker \bdy^T$.
\end{proof}

\section{Probability measures}
\label{sec:probability measures}

We are now ready to build probability measures on our spaces. 
\begin{definition}
\label{def:admissible everything}

We say that $\mu$ is an~\emph{admissible measure} on $\ED$ if $\mu$ is an $O(d)$-invariant finite Radon measure on $\ED$ with $\mu(\ED) > 0$ and finite first moment. 
\end{definition}
Recall that, while a purely measure-theoretic definition of Radon measure as a function on sets is standard, we can also view a Radon measure on $\R^n$ as a linear functional on the space $\mathcal{K}(\R^n)$ of continuous functions with compact support $\R^n \rightarrow \R$ via $\mu(f) = \mean{\mu}{f} = \int f(x) \mu(dx)$. A~\emph{probability measure} is a Radon measure with total mass $\mu(\R^n)$ one. A measure has~\emph{finite first moment} if the expected distance between two points is finite~\cite{Fritz2019}.

We note that every $O(d)$-invariant probability measure on $\ED$ with finite first moment is certainly admissible, but we are not requiring the total mass of the measure to be normalized to one. This is mostly a matter of notational convenience. If our network model involves only properties such as Gaussian springs, FENE potentials, and Lennard-Jones potentials, then $\mu$ is absolutely continuous with respect to Lebesgue measure and we have $\mu = p(W) \lambda^{d\edgesE}$ for some continuous density function $p(W)$. However, if our model involves equality constraints (such as fixed edgelengths), then $\mu$ may be singular.

We know that $\im \bdy^*$ is the subspace of $\ED$ of edge displacements which may actually be reassembled (via $\bdy^{*+}$) into vertex positions in a way that's compatible with the graph structure. Therefore, our ultimate goal is to condition $\mu$ on the hypothesis $W \in \im \bdy^*$. However, $\im \bdy^*$ is a measure-zero subset of $\ED$, and conditioning on such sets is not always well-defined. So we now introduce some (basically technical) constructions designed to ensure that we can build a conditional probability measure $\mu_\graphG$ supported on $\im \bdy^*$. Recall that we have already introduced the loop space (\defn{loop space}) $\ker \bdy \subset \EC$. We now introduce the corresponding subspace of $\ED$.
\begin{definition}
The~\emph{incompatible displacement} space $\ID \subset \ED$ is $(\im \bdy^*)^\perp = \ker \bdy^{*+}$. If $\ell \in \EC$ is in the loop space, we call $W(\ell)$~\emph{the failure to close} of $W$ around $\ell$.
\end{definition}
To motivate our second definition, suppose we can write $\ell = \edge_{\ell_1} + \cdots + \edge_{\ell_k}$ where without loss of generality we assume that the edges are oriented so that $\head \edge_{\ell_i} = \tail \edge_{\ell_{i+1}}$ and $\head \edge_{\ell_k} = \tail \edge_{\ell_ 1}$. Then since $W(\ell) = W(\edge_{\ell_1}) + \dots + W(\edge_{\ell_k})$, it is natural to think of $W(\ell)$ as the failure of the loop $\ell$ to close. More generally, \prop{loop space} tells us that every $\ell \in \ker \bdy$ is a linear combination of loops in this simple form, and hence $W(\ell)$ is the corresponding linear combination of failures to close around those simple loops.

A few facts about $\ID$ will be useful:
\begin{proposition}\label{prop:basic ID properties}
We have:
\begin{enumerate}
\item $\dim \ID = d \xi(\graphG)$, where $\xi(\graphG)$ is the cycle rank of $\graphG$. 
\item If $\ell \in \ker \bdy \subset \EC$ is a loop in $\graphG$, then $W(\ell) = (\proj_{\ID} W)(\ell)$.
\item $W$ is in $\im \bdy^*$ and hence $W = \operatorname{disp} X$ for some $X \in \VP$ if and only if the failure to close $W(\ell) = 0$ for all loops $\ell \in \ker \bdy$.
\end{enumerate}
\end{proposition}

\begin{proof}
We know from~\prop{ker and im of bdystar} that $\dim \im \bdy^* = d(\verticesV - 1)$. Therefore \begin{equation*}
\dim \ID = \dim (\im \bdy^*)^\perp = \dim \ED - \dim \im \bdy^* = d(\edgesE - \verticesV + 1) = d \xi(\graphG).
\end{equation*}
We know that $\ID \oplus \im \bdy^*$ is an orthogonal decomposition of $\ED$, so every $W \in \ED$ can be written as $W = W_{\ID} + W_{\im \bdy^*}$. But $(\im \bdy^*) = (\ker \bdy)^0$, so $W(\ell) = W_{\ID}(\ell)$, as required. The last claim follows immediately from $(\im \bdy^*) = (\ker \bdy)^0$ as well.
\end{proof}

We now recall a little background from probability theory:
\begin{notation}\label{pushforward notation}
	If $(S_1,\mathcal{A}_1)$ and $(S_2, \mathcal{A}_2)$ are Borel spaces, $f: S_1 \to S_2$ is measurable, and $\mu$ is a measure on $S_1$, then we will use $f_\sharp \mu$ to denote the pushforward measure on $S_2$; i.e., the measure defined by
	\[
		(f_\sharp\mu)(B) := \mu(f^{-1}(B)).
	\]
\end{notation}
The standard method in probability to construct conditional distributions is now to build a disintegration (cf.~\cite{Chang1997}) of $\mu$ relative to the map $\proj_{\ID}$ and Lebesgue measure on $\ID$. This construction yields conditional probability measures $\mugw$ concentrated on $\proj_{\ID}^{-1}(W)$ which are well-defined for~$\lambda^{d\xi(\graphG)}$-almost every ${W \in \ID}$. However, we would like to restrict our attention to cases where we can define a unique probability measure on $\im \bdy^* = \proj_{\ID}^{-1}(0)$, so we will require a slightly stronger idea originally proposed by Tjur~\cite{Tjur1975,Tjur1980}. We first give a version of Tjur's definition of a conditional probability which applies in the cases we study:
\begin{definition}[{Tjur~\cite[p.\ 6]{Tjur1975},~\cite[Sec.\ 9.7]{Tjur1980}}]\label{def:tjur conditional probability}
Suppose we have open sets $X \subset \R^m$ and $Y \subset \R^n$, a Radon probability measure $\mu$ on $X$, and a continuous map $t \co X \rightarrow Y$. For any $t_\sharp \mu$-measurable set $B \subset Y$ with $(t_\sharp \mu)(B) > 0$, we can define a Radon probability measure $\mu^{B}$ on $X$ by
\begin{equation*}
\mu^B(f) = \frac{1}{(t_\sharp\mu)(B)} \int_{t^{-1}(B)} f(x) \mu(dx)
\end{equation*}
for any $f \in \mathcal{K}(X)$.\footnote{Remember, $\mathcal{K}(X)$ is the space of continuous functions on $X$ with compact support.}

For any $y \in Y$, a measure $\mu^y$ on $X$ is~\emph{the conditional distribution of $\mu$ given $t(x) = y$} if, for any $f \in \mathcal{K}(\R^m)$ and any $\epsilon > 0$, there is an open neighborhood $V$ of $y$ in $Y$ so that, for any $B \subset V$ with $t_\sharp \mu(B) > 0$, we have $\abs{\mu^y(f) - \mu^B(f)} < \epsilon$. We note that if $\mu^y$ exists, it is unique. Further, it is concentrated on $t^{-1}(y)$.
\end{definition}
Intuitively, this definition says that $\mu^y$ (if it exists) is the (weak$^*$) limit of $\mu^B$ as the sets $B$ approach $y$. Tjur makes precise the notion of ``sets $B$ approaching a point $y$''~\cite[Definition 3.1]{Tjur1975}, but we don't need to worry about the details here. The observation that $\mu^y$ is unique if it exists is due to Tjur as well, so any reasonable definition of $\mu^y$ as a limit of $\mu^B$ yields the same result.
We can now define compatibility for one of our measures with a graph structure:
\begin{definition}\label{def:G-compatible}
We say that $\mu$ is~\emph{compatible with $\graphG$} if $\mu$ is admissible and there is an open ball $U \subset \ID$ centered at $0$ so that the conditional distributions $\mugw := \mu^W$ constructed using $\mu$ as the Radon measure on $\ED$ and $\proj_{\ID}$ as the continuous map $\ED \rightarrow \ID$ are defined for all $W \in U$. We define $\mu_{\ID} := (\proj_{\ID})_\sharp \mu$ and $\mug := \mugz$. 
\end{definition}
This definition is certainly straightforward, but as written it may seem difficult to check. The following~Proposition shows that compatibility is automatic (and the conditional distributions have a familiar form) in a wide variety of cases where $\mu$ is given in terms of a density function.
\begin{proposition}\label{prop:disintegration with density}
Suppose that $\mu$ on $\ED$ is $O(d)$-invariant and has a continuous density $p(Z)$ with respect to Lebesgue measure $\lambda^{d\edgesE}$ on $\ED$. Further, suppose there are open sets $U' \in \ED$ and $U \in \ID$ so that $\proj_{\ID}(U') = U$, $0 \in U$, and $p(Z) > 0$ on $U'$. Last, suppose $p(Z) \in o(\norm{Z}^n)$ for all $n \in \Z$. 

Then $\mu$ is admissible and compatible with $\graphG$, $\mu_{\ID}$ has a continuous, positive density with respect to Lebesgue measure $\lambda^{d\xi(\graphG)}$ on $U$ given by
\begin{equation*}
m_W = \int_{Z \in \proj_{\ID}^{-1}(W)} p(Z) \haus{d(\verticesV-1)}(dZ),
\end{equation*}
and the conditional distributions $\mugw$ (in the sense of~\defn{tjur conditional probability}) for all $W \in U$ are given explicitly by 
\begin{equation*}
\mugw(f) = \frac{1}{m_W} \int_{Z \in \proj_{\ID}^{-1}(W)} f(Z) p(Z) \haus{d(\verticesV-1)}(dZ).
\end{equation*}
\end{proposition}
Above, $\haus{d(\verticesV-1)}$ is the Hausdorff (or surface) measure on the $d(\verticesV-1)$-dimensional subspace $\proj_{\ID}^{-1}(W)$ of $\ED$.

Our hypotheses require that any neighborhood of $0$ in $\ID$ is assigned a positive probability by $\mu_{\ID}$ (that is, there is a nonzero chance of random configurations of edge displacements which are nearly consistent with the graph $\graphG$), which certainly seems reasonable. We note that if our model implies a mixture of lower and upper bounds on distances between vertices, this hypothesis can fail to be satisfied in somewhat subtle ways (all of the generalized triangle inequalities on these distances must be able to be satisfied), so it's difficult to make a more general statement about when this happens.

Further, we require the decay condition $p(Z) \in o(\norm{Z}^n)$ for all $n \in \Z$, which states that $p(Z)$ decays faster than any polynomial as $\norm{Z} \rightarrow \infty$. This will ensure continuity of the $m_W$ by preventing $m_{W_i} \rightarrow \infty$ as $W_i \rightarrow W$. We note that this is automatic when $p(Z)$ has compact support, and that it could be weakened to polynomial decay for a particular $n$ depending on $\graphG$.

This is very close to the definition of conditional probability in terms of marginal densities that is usually given in textbooks; we are simply giving alternate hypotheses which ensure that the marginal density $m_W$ is everywhere defined, positive, and continuous. 
\begin{proof}
We note that Lemma~2.5 of~\cite{Fritz2019} observes that $\mu$ has finite first moment $\iff$ there is some $Z_0 \in \ED$ so that the $\mean{\mu}{\norm{Z-Z_0}}$ is finite. Our decay condition on $p(Z)$ implies this for $Z_0 = 0$, and also implies that the total mass $\mean{\mu}{1}$ is finite. Of course, any measure with a continuous density with respect to Lebesgue measure is Radon.

Theorem 8.1 of~\cite{Tjur1975} does almost all the work here (recalling that $\proj_{\ID}$ is an orthogonal projection, so its normal Jacobian is constant and one); the only thing we have to prove is that $m_W$ is a continuous, positive function of $W$. 

We note that for any $W \in U$, there is some open $\proj_{\ID}^{-1}(W) \cap U'$ where $p(Z) > 0$. Therefore, the integral $m_W$ of the (everywhere non-negative) $p(Z)$ over $\proj_{\ID}^{-1}(W)$ is positive.

We now show continuity. Suppose we have $W_i \rightarrow W$ in $\ID$. The subspaces $\proj_{\ID}^{-1}(W_i)$ are all in the form $\im \bdy^* \oplus W_i$. For each $Z \in \im \bdy^*$, we can define $q_i(Z) := p(W_i + Z)$ and $q(Z) := p(W + Z)$. Now
\begin{equation*}
m_{W_i} = \int_{Z \in \proj_{\ID}^{-1}(W_i)} p(Z) \haus{d(\verticesV-1)}{dZ} = \int_{Z \in \im \bdy^*} q_i(Z) \haus{d(\verticesV-1)}{dZ}.
\end{equation*}

Since $p$ is continuous on $\ED$ and $W_i \rightarrow W$, we have $q_i(Z) \rightarrow q(Z)$ pointwise. By the Lebesgue dominated convergence theorem, to show that the integrals $m_{W_i} \rightarrow m_W$ it now suffices to show that the $q_i$ are dominated by some integrable function $g$ on $\im\bdy^*$. 

Since the $q_i$ are shifts of the continuous function $p$ by $W_i$ which lie in a bounded subset of $\ED$ (because the $W_i$ converge), we may assume that the $q_i$ are uniformly bounded on any compact subset of $\im\bdy^*$. By the same logic, since $p(Z)/\norm{Z}^n \rightarrow 0$ as $\norm{Z} \rightarrow \infty$ for any $n$, we may assume for any $n$ that there is a single ball $B_n \subset \im\bdy^*$ so that $q_i(Z) < \norm{Z}^n$ for $Z$ outside $B_n$. 

We've already argued that there is some $C_n > q_i$ on $B_n$, so we can now construct a function $g_n(Z)$ equal to $C_n$ inside $B_n$ and $\norm{Z}^n$ outside $B_n$ and observe that $g_n > q_i$. If we take $n < 0$ so that $\abs{n}$ is sufficiently large, $g_n$ will be integrable on $\im \bdy^*$ (that is, $\int_{\im \bdy^*} g_n(Z) \haus{d(\verticesV-1)}{dZ} < \infty$), completing the proof. 
\end{proof}
It's sometimes easier to use these alternate hypotheses: 
\begin{corollary}\label{cor:compatibility for joint distributions}
Suppose that we have $O(d)$-invariant probability distributions $\rho_1, \dots, \rho_\edgesE$ on the edges $\edge_1, \dots, \edge_\edgesE$ of $\graphG$ and further suppose that each $\rho_i$ has a continuous density with respect to Lebesgue measure $\rho_i(dx) = p_i(\norm{x}) \lambda^{d}$ on $\R^d$ and that each $p_i$ is bounded, positive in a neighborhood of $0$, and has $p_i(x) \in o(\norm{x}^n)$ for all $n \in \Z$. 

If $\mu$ is the joint distribution of independent edge displacements sampled from the $\rho_i$, then $\mu$ has density
\begin{equation*}
p(Z) = p_1(\norm{Z(\edge_1)}) \,\cdots\, p_\edgesE(\norm{Z(\edge_\edgesE)}).
\end{equation*} 
with respect to Lebesgue measure $\lambda^{d\edgesE}$ on $\ED$ and the conclusions of~\prop{disintegration with density} hold.
\end{corollary}

\begin{proof}
As in the proof of~\prop{disintegration with density} we need only show that each 
\begin{equation*}
m_W = \int_{Z \in \proj_{\ID}^{-1}(W)} p(Z) \haus{d(\verticesV - 1)}(dZ) 
\end{equation*}
is positive and continuous in a neighborhood of $0$ in $\ID$. 

To see that $m_W$ is positive for small enough $W$, observe that $S(W) := \proj_{\ID}^{-1}(W)$ is an affine subspace of $\ED$ whose closest point is $\norm{W}$ from $0$. Therefore, by choosing $W$ small enough, we can guarantee that $S(W)$ intersects any open neighborhood of $0$ in $\ED$ in a set of positive measure. 

Now we know that 
\begin{equation*}
p(Z) = p_1(\norm{Z(\edge_1)}) \,\cdots\, p_\edgesE(\norm{Z(\edge_\edgesE)}).
\end{equation*}
It's an exercise to show that 
\begin{equation}\label{eq:ri bounds}
\frac{1}{\sqrt{\edgesE}} \norm{Z} \leq \max_i \norm{Z(\edge_i)} \leq \norm{Z}.
\end{equation}
In particular, since we've assumed that each $p_i$ is positive in a neighborhood of $0$, there is some $B>0$ so that all $p_i(\norm{Z(\edge_i)})$ are positive for $\norm{Z(\edge_i)} \leq \norm{Z} \leq B$.
Thus for any $W$ with $\norm{W} < B$, $p(Z)$ is positive on subset of $S(W)$ with positive measure. This establishes that these $m_W$ are positive.

Since the $p_i$ are all bounded, without loss of generality they are all bounded by a common $B>0$. It follows from~\eqn{ri bounds} that for each $Z$ there is some $i$ so that
\begin{equation*}
p(Z) \leq B^{\edgesE-1} p_i\left(\frac{1}{\sqrt{\edgesE}} \norm{Z}\right). 
\end{equation*}
Since each $p_i(\norm{x})$ is in $o(\norm{x}^{n})$ (and there are a fixed number of $p_i$), this implies that $p(Z)$ is in $o(\norm{Z}^{n})$. The remainder of the argument follows as in the proof of~\prop{disintegration with density}.
\end{proof}
Most of the models we'd like to consider fall under~\cor{compatibility for joint distributions}:
\begin{corollary}\label{cor:phantom network theory works}
Suppose we have the Gaussian phantom network model of James--Guth, where $\mu$ is the joint distribution of edge displacements distributed according to any mean-zero Gaussians on $\R^d$. Then $\mu$ is admissible and compatible with $\graphG$. 
\end{corollary}
\begin{corollary}
Suppose that $\mu$ is the joint distribution of independent edge displacements distributed according to any Boltzmann distributions $p_i(x) \sim \exp( -f_i(x) )$ where the energy functions $f_i(x)$ are in $O(\norm{x}^\alpha)$ for some positive $\alpha$ as $\norm{x} \rightarrow \infty$. Then $\mu$ is admissible and compatible with $\graphG$.
\end{corollary}
Now that we have shown that $\mu$ is compatible with $\graphG$ often enough to make~\defn{G-compatible} interesting, we establish some properties of the $\mugw$.
\begin{proposition}\label{prop:disintegration}
If $\mu$ is compatible with $\graphG$, we have for each $W$ in $U$ that $\mugw$ is concentrated on $\proj_{\ID}^{-1}(W)$, and $\mug$ is concentrated on $\im \bdy^* = \proj_{\ID}^{-1}(0)$.

Further, if $Q \in O(d)$, then $\mu_{\ID} = Q_\sharp \mu_{\ID}$, $Q_\sharp \mugw = \mugqw$ and $Q_\sharp \mug = \mug$.
\end{proposition}

\begin{proof}
Proving this requires us to introduce Tjur's idea of a~\emph{decomposition} of a measure with respect to a map, which is like a disintegration~(cf.\, \cite{Chang1997}) but somewhat stronger:
\begin{definition}[{\cite[Definition 6.1]{Tjur1975}}]\label{defn:decomposition}
Let $t \co X \rightarrow Y$ be a continuous function and $\lambda$ be a measure on $X$. A family $\lambda_y$ (for $y \in Y$) and a measure $\lambda'$ on $Y$ are called a~\emph{decomposition of $\lambda$ with respect to $t$} if 
\begin{enumerate}
\item The mapping $y \mapsto \lambda_y$ is continuous (in the weak$^*$ topology on measures),
\item Each $\lambda_y$ is concentrated on $t^{-1}(y)$.
\item For any $f \in \mathcal{K}(X)$, $\int \lambda_y(f) \lambda'(dy) = \lambda(f)$.
\end{enumerate}
\end{definition}
Conditional distributions are connected to decompositions very tightly by the following:
\begin{theorem}[{\cite[Theorem 7.1]{Tjur1975}}]\label{thm:decomposition and existence}
Let $X \subset \R^m$ and $Y \subset \R^n$ be open sets, $t \co X \rightarrow Y$ be a continuous map, and $\mu$ be a probability measure on $X$. Suppose that $\lambda_y$ ($y \in Y$) is a family of probability measures on $X$. A set of $\lambda_y$ and $t_\sharp \mu$ are a decomposition of $\mu$ with respect to $t$ $\iff$ the conditional distribution $\mu^y$ is defined for all $y \in Y$ and $\mu^y = \lambda_y$.
\end{theorem}
By~\thm{decomposition and existence}, $\mu$ is compatible with $\graphG$ $\iff$ the $\mugw$ and $\mu_{\ID}$ are a decomposition of $\mu$ with respect to $\proj_{\ID}$. This implies first that each $\mugw$ is concentrated on $\proj_{\ID}^{-1}(W)$. 

Because $\mu$ is $\graphG$-compatible, it is admissible, and therefore $O(d)$-invariant: for any $Q \in O(d)$, $Q_\sharp \mu = \mu$. We now show that this implies that $Q_\sharp \mu_{\ID} = \mu_{\ID}$.
The action of $Q$ on $\ED$ is multiplication by the $d\edgesE \times d\edgesE$ matrix $Q \otimes I_{\edgesE}$. As a matrix, $\proj_{\ID} = I - \bdy^* \bdy^{*+} = I_d \otimes (I_\edgesE - \bdy^{T} \bdy^{T+})$. It follows that 
\begin{equation*}
Q \proj_{\ID} 
= (Q \otimes I_{\edgesE})(I_\edgesE - \bdy^{T} \bdy^{T+}) 
= Q \otimes (I_\edgesE - \bdy^{T} \bdy^{T+}) 
= (I_d \otimes (I_\edgesE - \bdy^{T} \bdy^{T+}))(Q \otimes I_{\edgesE}) 
= \proj_{\ID} Q.
\end{equation*}
where we used \lem{mixed product} in the middle steps. Now we know that 
\begin{equation*}
Q_\sharp \mu_{\ID} = Q_\sharp (\proj_{\ID})_\sharp \mu = (\proj_{\ID})_\sharp Q_\sharp \mu = (\proj_{\ID})_\sharp \mu = \mu_{\ID}.
\end{equation*}
We now prove $\mugw = Q_\sharp \mugw$, following the lines of Example 7 in~\cite{Chang1997}. Fix some $Q \in O(d)$. By~\thm{decomposition and existence}, it suffices to show that $\mugqw$ and $Q_\sharp \mugw$ are both decompositions of $\mu$ with respect to $Q^{-1} \circ \proj_{\ID}$.

We start with the $\mugqw$. We already know that $W \mapsto \mugw$ is (weak$^*$) continuous in $W$, so $Q(W) \mapsto \mugqw$ is as well. Since $W \mapsto Q(W)$ is also continuous, the composition $W \mapsto \mugqw$ is continuous. 

Since each of the measures $\mugw$ is concentrated on $\proj_{\ID}^{-1}(W)$, the measure $\mugqw$ is concentrated on $\proj_{\ID}^{-1}(Q(W)) = (Q^{-1} \circ \proj_{\ID})^{-1}(W)$. 

Now suppose we have some $f \in \mathcal{K}(\ED)$. We have
\begin{align*}
\int_{\ED} f(Z) \mu(dZ) &= \int_{U} \mugw(f) \, \mu_{\ID}(dW) = \int_{U} \mugw(f) \, Q_\sharp \mu_{\ID}(dW) = \int_{U} \mugqw(f) \, \mu_{\ID}(dW).
\end{align*}
We have now established that the $\mugqw$ are a $Q^{-1} \circ \proj_{\ID}$-decomposition of $\mu$. 

We now show the same for the $Q_\sharp \mugw$, but it will help to prove this as a Lemma about decompositions in general, because we'll repeat similar arguments later.
\begin{lemma}\label{lem:decompositions push forward}
Suppose $t \co X \rightarrow Y$ is a continuous function and $\lambda$ is a measure on $X$, and further suppose we have a family of measures $\lambda_y$ (for $y \in Y$) and a measure $\lambda'$ on $Y$ which are a decomposition of $\lambda$ with respect to $t$. 

Further, suppose that we have a Borel map $g \co X \rightarrow Z$ and a map $s \co Z \rightarrow Y$ so that $s \circ g = t$. Then the pushforwards $g_\sharp \lambda_y$ and the measure $\lambda'$ on $Y$ are a decomposition of $g_\sharp \lambda$ with respect to $s$.
\end{lemma}
\begin{proof}[Proof of lemma]
Since the $y \mapsto \lambda_y$ is a continuous map from $y$ to Radon measures, and $\mu \mapsto g_\sharp \mu$ is a continuous map between measures, the composition $y \mapsto g_\sharp \lambda_y$ is continuous. 

Suppose we have some open set $U \in Z$ so that $U \cap s^{-1}(y) = \emptyset$. We claim that $g_\sharp \lambda_y (U) = 0$. We know that $g_\sharp \lambda_y (U) = \lambda_y (g^{-1}(U))$. We observe that $t(g^{-1}(U)) = s(g(g^{-1}(U))) \subset s(U)$. In particular, $y \not\in t(g^{-1}(U))$ since $y \notin s(U)$. Thus $\lambda_y (g^{-1}(U)) = 0$, as desired.

Suppose we have $f \in \mathcal{K}(Z)$. Then 
\begin{equation*}
g_\sharp \lambda(f) = \lambda (f \circ g) = \int \lambda_y(f \circ g) \lambda'(dy)
= \int (g_\sharp \lambda_y) (f) \lambda'(dy),
\end{equation*}
as desired.
\end{proof}
Now if we let $Q \co \ED \rightarrow \ED$ and $\proj_{\ID} \co \ED \rightarrow \ID$, we see that $Q^{-1} \circ \proj_{\ID}$ has the property that $Q \circ (Q^{-1} \circ \proj_{\ID}) = \proj_{\ID}$, and so by ~\lem{decompositions push forward}, the $Q_\# \mugw$ are a decomposition of $\mu$ with respect to $(Q^{-1} \circ \proj_{\ID})$, as desired.
\end{proof}

We are now going to examine two singular measures $\mu$; one where we can establish compatibility and one where we can see that compatibility fails. In these cases, our arguments will be much more specific to the model. We are first going to recall a key fact about the freely jointed chain:
\begin{proposition}\label{prop:conditional probability for arms}
If $\edgesE \geq 3$, $X = \R^{3\edgesE}$, and $\mu$ is the product of uniform area measures $\mu_i$ on the unit spheres $S^2 \subset \R^3$, $Y = \R^3$, and $\ftc \co X \rightarrow Y$ is the vector sum $\ftc(x) = x_1 + \cdots + x_\edgesE = y$, then there are well-defined conditional probabilities $\mu^y$ for each $y$ with $\norm{y} < \edgesE$ and a measure 
\begin{equation}\label{eq:failure to close for equilateral edges}
\begin{aligned}
\ftc_\sharp \mu(y) &= \left(\frac{1}{2\pi^2 \ell} \int_0^\infty s \sin \ell s \sinc^\edgesE s \,d{s} \right) \lambda^{3}(dy) \\
 &= \left( \frac{\edgesE-1}{2^{\edgesE+1}\pi\ell}\sum_{k=0}^{\edgesE-1}\frac{(-1)^k }{k!(\edgesE-k-1)!}  \left((\edgesE+\ell-2k-2)_+^{\edgesE-2}-(\edgesE+\ell-2k)_+^{\edgesE-2}\right)  \right) \lambda^{3}(dy)
\end{aligned}
\end{equation}
where $x_+ = \max \{ x, 0\}$ and $\ell = \norm{y}$. These also form a decomposition of $\mu$ with respect to $\ftc$.
\end{proposition}

\begin{proof}
The computation of the pushforward density and the expression as a $\sinc$ integral can be traced back to Rayleigh~\cite{Rayleigh:1919do}. The existence of the conditional probabilities is more or less standard, but we outline the argument in order to connect it to decompositions explicitly.  

We start by defining a singular set $\Delta \subset (S^2)^\edgesE$ as the set of all configurations where $x_i = \pm x_j$ for all $i, j \in 1, \dots, \edgesE$. This is a finite union of submanifolds of dimension $2$ inside $(S^2)^\edgesE$ and so we note that $\haus{k}(\Delta) = 0$ for any $k > 2$. 

We now restrict our attention to $(S^2)^\edgesE - \Delta$. The map $(z,\theta) \mapsto (\sqrt{1 - z^2} \cos \theta, \sqrt{1 - z^2} \sin \theta, z)$ writing $S^2$ in cylindrical coordinates is area-preserving and so pushes forward Lebesgue measure $\lambda^2(dz d\theta)$ to the area measure on $S^2$. These particular coordinates don't cover the north and south poles $(0,0,\pm 1)$, but employing similar constructions with respect to $x$ and $y$ and a partition of unity, we may construct a finite number of coordinate patches $\cyl_k \co Z_k \rightarrow X_k$ where the $Z_k$ are an open cover of $(-1,1)^\edgesE \times (0,2\pi)^\edgesE$, the $X_k$ are an open cover of $(S^2)^\edgesE$, the $\cyl_k$ are surjective, and each patch comes with a smooth, positive density function $\alpha_k$ so $\sum (\cyl_k)_\sharp \alpha_k \lambda^{2 \edgesE} = \mu$. We define $\Delta_k := (\cyl_k)^{-1}(\Delta)$, and note that $\Delta_k$ is some dimension $2$ submanifold of $Z_k$.

We now consider the maps $\ftc \circ \cyl_k \co Z_k \rightarrow \R^3$. Since $\cyl_k$ is a diffeomorphism, its differential is always invertible. But the differential of $\ftc$ at $x$ is not invertible if and only if all the $x_i$ are colinear: exactly when $x \in \Delta$. Combining these, we see that the differential of $\ftc \circ \cyl_k$ is surjective on $Z_k - \Delta_k$. The normal Jacobian $J_3(\ftc \circ \cyl_k) = \det( D\ftc D\cyl_k D\cyl_k^T D\ftc^T )^{1/2}$ is therefore positive on $Z_k - \Delta_k$.

By Theorem 8.1 of~\cite{Tjur1975}, we can construct decompositions $\alpha_k^y$ of each $\alpha_k  \lambda^{2\edgesE}$ with respect to $\ftc \circ \cyl_k$ on the open sets $Z_k - \Delta_k$, where
\begin{equation*}
q_k(y) = \int\limits_{z \in (\ftc \circ \cyl_k)^{-1}(y) - \Delta_k} \frac{\alpha_k(z)}{J_3 (\ftc \circ \cyl_k)(z)} \haus{2\edgesE-3}(dz)
\end{equation*}
and\footnote{Here the notation $\nu \restricted U$ means the restriction of the measure $\nu$ to the subset $U$.} 
\begin{equation*}
\alpha_k^y = \frac{1}{q_k(y)} \frac{\alpha_k(z)}{J_3 (\ftc \circ \cyl_k)(z)} (\haus{2\edgesE-3} \restricted ((\ftc \circ \cyl_k)^{-1}(y) - \Delta_k))(dz)
\end{equation*}
and the pushforward measure $(\ftc \circ \cyl_k)_\sharp \alpha_k \lambda^{2\edgesE}$ has density $q_k(y)$ with respect to $\lambda^{3}$ as long as the $q_k(y)$ are continuous and positive in $y$. The $q_k$ are integrals of positive quantities and hence positive; to see that they are finite recall that $\sum_k (\cyl_k)_\sharp \alpha_k \lambda^{2\edgesE} = \mu$, so $\sum_k (\ftc \circ \cyl_k)_\# \alpha_k \lambda^{2\edgesE} = \ftc_\# \mu$, which is given by the finite positive density in~\eqn{failure to close for equilateral edges}.

We now want to assemble our work. We can define a map $\cyl \co \bigsqcup_k (Z_k - \Delta_k) \rightarrow (S^2)^\edgesE - \Delta$ by letting the restriction of $\cyl$ to $Z_k - \Delta_k$ be $\cyl_k$ and a measure $\alpha$ on $\bigsqcup_k (Z_k - \Delta)$ whose restriction to $Z_k - \Delta_k$ is $\alpha_k \lambda^{2\edgesE}$. If we also define $\alpha^y$ on $\bigsqcup_k (Z_k - \Delta)$ by letting its restriction on $Z_k - \Delta$ be $\alpha_k^y$, it then follows that the $\alpha^y$ are a decomposition of $\alpha$ with respect to $\ftc \circ \cyl$. Now we have noted above that $\haus{\alpha}(Z_k) = 0$ for any $\alpha > 2$. Since $\edgesE \geq 3$, we see that $\haus{2\edgesE-3}(Z_k) = 0$ and we may rewrite our decomposition as 
\begin{equation*}
q_k(y) = \int\limits_{z \in (\ftc \circ \cyl_k)^{-1}(y)} \frac{\alpha_k(z)}{J_3 (\ftc \circ \cyl_k)(z)} \haus{2\edgesE-3}(dz)
\end{equation*}
and 
\begin{equation*}
\alpha_k^y = \frac{1}{q_k(y)} \frac{\alpha_k(z)}{J_3 (\ftc \circ \cyl_k)(z)} (\haus{2\edgesE-3} \restricted (\ftc \circ \cyl_k)^{-1}(y))(dz)
\end{equation*}
without changing it, as long as (by convention) we replace the integrand by $1$ where it is not defined. Therefore the $\alpha^y$ are~\emph{also} a decomposition of $\alpha$ with respect to $\ftc \circ \cyl \co \bigsqcup_k Z_k \rightarrow (S^2)^\edgesE$. We would like to call attention to this step because it is the only one in the proof which is nonstandard-- in a generic situation, one would know from Sard's theorem that the singular set $Z$ had $\haus{2\edgesE}$-measure zero, so it could be ignored when computing expectations over all of $(S^2)^\edgesE$ with respect to $\mu$, but one would~\emph{not} then be able to conclude that you could ignore the singular set when computing expectations with respect to~\emph{conditional} probabilities, which are integrals over lower-dimensional spaces.

We know that $\mu = \sum (\cyl_k)_\sharp \alpha_k \lambda^{2 \edgesE} = \cyl_\sharp \alpha$. We define $\mu^y = \cyl_\sharp \alpha^y$. It follows from~\lem{decompositions push forward} that the $\mu^y$ are a decomposition of $\mu$ with respect to $\ftc \co (\R^3)^\edgesE \rightarrow \R^3$.
\end{proof}
We have an immediate corollary:
\begin{corollary}\label{cor:freely jointed ring admissible}
 Let $\edgesE \geq 3$, suppose $\graphG$ is the $\edgesE$-edge cycle graph, and let $\mu$ be the product of (uniform) area measures on the product of unit spheres $(S^2)^\edgesE \subset \ED = (\R^3)^\edgesE$. Then $\mu$ is admissible and compatible with~$\graphG$.
\end{corollary}
\begin{proof}
We note that the fact that $\mu$ has finite mass and finite first moment follows directly from the fact that $\mu$ has compact support.
\end{proof}
By contrast, suppose we took $\graphG$ to be the $\edgesE$-cycle graph and let $\mu$ be supported on the intersection of $(S^2)^\edgesE$ and $\norm{W(\ell_1)} = 1$. The measure $\mu$ would still be Radon and $O(3)$-invariant, and therefore admissible. But the pushforward measure $\mu_{\ID}$ would be supported on the unit sphere $\norm{W(\ell_1)} = 1$ (and by $O(3)$-invariance, be the area measure on that sphere). Thus we cannot define conditional probabilities for $W(\ell_1)$ near $0$ and this $\mu$ is~\emph{not} compatible with $\graphG$.

As we pointed out after the proof of Proposition~\ref{prop:ker and im of bdystar}, collections of vertex positions have the same edge displacements if and only if they are related by a translation. This means that the map $\bdy^{*+}$, which reconstructs vertex positions from edge displacements, must choose a particular translation of the vertices. We now show that this choice is natural.

\begin{definition}\label{def:centered}
We say that $X \in \VP$ is a~\emph{centered} embedding of $\graphG$ if the position of the center of mass $\frac{1}{\verticesV} X(\vertex_1 + \cdots + \vertex_\verticesV) = 0$.
\end{definition}

\begin{proposition}\label{prop:centered}
The space $\im \bdy^{*+}$ is the space of centered embeddings. We also have
\begin{equation*}
\{ \text{centered embeddings} \} = (\im \onesVV)^0 = \ker \onesVV^* = (\ker \bdy^T)^0 = \im \bdy^{T*}
 = (\ker \bdy^*)^\perp.
 \end{equation*} 
\end{proposition}

\begin{proof}
If $X$ is centered,~\defn{centered} tells us that $X(y) = 0$ for every $y \in \im \onesVV$. Using~\prop{annihilator props}, the centered embeddings are then $(\im \onesVV)^0 = \ker \onesVV^*$. Since $\im \onesVV = \ker \bdy^T$ by~\prop{basic Laplacian properties}, they are also $(\ker \bdy^T)^0 = \im \bdy^{T*} = \im \bdy^{*+}$ .
\end{proof}

This proposition gives a formal explanation for choosing centered configurations of vertices. We now explain why this is also a physically natural choice to make. A probability distribution for configurations of a polymer in the absence of an external field should not depend on the position or orientation of the polymer in space. We have ruled out dependence on orientation by insisting that our probability distributions be $O(d)$-invariant. 

To rule out dependence on translation we must\footnote{We cannot simply insist on a translation-invariant probability measure because every translation-invariant Radon measure on a finite-dimensional vector space has infinite or zero total mass and so cannot be a probability measure.} restrict our attention, and our probability distribution, to vertex positions in a particular subspace $\mathcal{S}$ of $\VP$. Every $X$ in $\VP$ must have a translation which lies in $\mathcal{S}$ and no two $X$ and $Y$ in the subspace can be related by a translation (formally, $\mathcal{S}$ is a cross-section of the action of the translation group on $\VP$). In other words, $\VP$ is the direct sum of the translation subspace $\ker \bdy^*$ and the transverse subspace $\mathcal{S}$: $\VP = \ker \bdy^* \oplus \mathcal{S}$. 

An obvious choice for $\mathcal{S}$ is the orthogonal complement $(\ker \bdy^*)^\bot = \im \bdy^{*+}$ of the translations. By \prop{centered}, these are the centered configurations where we put the center of mass at the origin, as in Eichinger~\cite{Eichinger1972}.\footnote{We could also have fixed the position of a vertex, as in James--Guth theory, or fixed any other linear combination of vertex positions. All possible choices of $\mathcal{S}$ are linearly isomorphic, and indeed the isomorphism can be realized as orthogonal projection, so it is straightforward to translate between different conventions.} This is physically reasonable since polymer networks are usually very large, so the fluctuations of the center of mass are small. Further, we have found that computations of the response of a polymer to an external force are simplified for centered configurations.

We now verify that every physically meaningful probability distribution on centered embeddings corresponds to a distribution on edge displacements.
 
\begin{proposition} \label{prop:always a nug}
If $\nug$ is any probability measure on $\VP$ which is $O(d)$-invariant and concentrated on the centered configurations $\im \bdy^{*+}$, then $\nug = (\bdy^{*+})_\sharp \mug$ for some $O(d)$-invariant $\mug$ concentrated on $\im \bdy^*$.
\end{proposition}

\begin{proof}
Since the centered configurations are $\im \bdy^{*+}$, the map $\bdy^{*+} \bdy^* = \proj_{\im \bdy^{*+}}$ fixes the centered configurations. Since $\nug$ is concentrated on the centered configurations, this means that 
\begin{equation*}
\nug = (\bdy^{*+} \bdy^*)_\sharp \nug = (\bdy^{*+})_\sharp (\bdy^*_\sharp \nug).
\end{equation*}
Set $\mug = \bdy^*_\sharp \nug$. Since $\mug$ is the pushforward of an $O(d)$-invariant measure by the $O(d)$-equivariant function $\bdy^*$ (see \lem{linear equivariance}), it is $O(d)$-invariant. 
\end{proof}
Incidentally, this solves a computational problem of some interest in numerical experiments: if one is given a set of edge displacements $W = \operatorname{disp}(X) = \bdy^* X$, this shows that you can reconstruct $X$ by applying the pseudoinverse matrix $\bdy^{*+}$. Koohestani and Guest~\cite{Koohestani2013} use this same idea to reconstruct conformations of tensegrities by imposing loop conditions.

\begin{algorithm}
If $W \in \ED$ satisfies $W = \operatorname{disp}(X)$ then $X = \bdy^{*+} W$ is the unique centered $X \in \VP$ with $\operatorname{disp} X = \bdy^* X = W$. However, computing $\bdy^{*+}$ as a $d\verticesV \times d\edgesE$ matrix is awkward because it requires us to choose bases for $\VP$ and $\ED$ and write $X$ and $W$ as vectors. It is more convenient to write $X$ and $V$ as $d \times \verticesV$ and $d \times \edgesE$ matrices, respectively, so that $\mat{X} = \mat{W} \mat{\bdy^+}$ or $\mat{X}^T = \mat{\bdy^{T+}} \mat{W}^T$. Then we need only find the $\edgesE \times \verticesV$ matrix $\mat{\bdy^{+}}$ to solve the embedding problem for $\graphG$. The amount of work required to do so does not depend on the embedding dimension $d$.
\end{algorithm}

\section{Means and Variances}
\label{sec:means and variances}

We would now like to draw whatever conclusions we can about the means and variances of vertex positions and edge displacements for a generic $\mug$ and $\nug$ from $O(d)$-invariance and concentration on $\im \bdy^*$ and $(\ker \bdy^*)^\perp$.
\begin{lemma}\label{lem:physicality implies mean zero}
If $\nu$ is any $O(d)$-invariant probability measure on $\VP$ then $\mean{\nu}{X} = 0$. If $\mu$ is any $O(d)$-invariant probability measure on $\ED$, then $\mean{\mu}{W} = 0$.
\end{lemma}

\begin{proof}
Suppose we have a $d \times d$ matrix $Q \in O(d)$. The action of $Q$ on $\VP$ is multiplication by the $d\verticesV \times d\verticesV$ matrix $Q \otimes I_\verticesV$. Since $\nu$ is $O(d)$-invariant, we know $(Q \otimes I_\verticesV)_\sharp \nu = \nu$. But then
\begin{equation*}
\mean{\nu}{X} = \mean{(Q \otimes I_\verticesV)_\sharp \nu}{(Q \otimes I_\verticesV) X} = \mean{\nu}{(Q \otimes I_\verticesV) X} = (Q \otimes I_\verticesV) \mean{\nu}{X}.
\end{equation*}
This is true for all $Q \in O(d)$ only if $\mean{\nu}{X} = 0$. The proof for $\mean{\mu}{W}$ is the same.
\end{proof}

We recall that if $x$ is a random vector chosen according to a probability measure $\rho$ on an inner product space $\Vprod$, then by definition
\begin{equation*}
\cov{\rho}{u}{v} := \mean{\rho}{\Vp{x}{u} \Vp{x}{v}}. 
\end{equation*}
If the inner product $\Vp{-}{-} = \dotp{-}{-}_A$ for a symmetric matrix $A$, then 
\begin{equation}
\begin{aligned}
\cov{\rho}{u}{v} &= \mean{\rho}{\dotp{x}{Au} \dotp{x}{Av}} = \mean{\rho}{\dotp{Ax}{u} \dotp{Ax}{v}} 
= \mean{\rho}{u^T Axx^TA^T v} \\
&= u^T\mean{\rho}{Axx^TA^T}v = \dotp{u}{\mean{\rho}{Axx^TA} v} =
\dotp{u}{A \mean{\rho}{xx^t} A v}.
\end{aligned}
\label{eq:covariance in general inner product}
\end{equation}
Thus we really need to compute the expectations of the outer products $X X^T$ and $W W^T$ to understand the covariances of vertex positions and edge vectors. 

\begin{proposition}
\label{prop:covariance structure for VP}
If $\nu$ is any $O(d)$-invariant probability measure on $\VP$, then the $d\verticesV \times d\verticesV$ matrix $\mean{\nu}{X X^T} = I_d \otimes \Sigma_\verticesV$ where $\Sigma_\verticesV$ is a symmetric positive semidefinite $\verticesV \times \verticesV$ matrix giving the expectations of products of \emph{a fixed coordinate} of the positions of different vertices of $\graphG$. 

It follows immediately that if we view $X$ as a $d \times \verticesV$ matrix, the $d \times d$ matrix of expected dot products (in $\R^\verticesV$) of coordinates of vertex positions $\mean{\nug}{XX^T} = (\tr \Sigma_\verticesV) I_d$ and the $\verticesV \times \verticesV$ matrix of expected dot products (in $\R^d$) of vertex vectors $\mean{\nug}{X^TX} = d \Sigma_\verticesV$.
\label{prop:block structure of XXt}
\end{proposition}

\begin{proof}
As above, 
\begin{align*}
\mean{\nu}{XX^T} &= \mean{(Q \otimes I_\verticesV)_\sharp \nu}{(Q \otimes I_\verticesV) X X^T (Q^T \otimes I_\verticesV)} \\
&= \mean{\nu}{(Q \otimes I_\verticesV) X X^T (Q^T \otimes I_\verticesV)} 
= (Q \otimes I_\verticesV) \mean{\nu}{X X^T} (Q^T \otimes I_\verticesV).
\end{align*}
But if $\mean{\nu}{X X^T}$ is invariant under conjugation by $Q \otimes I_\verticesV$, as a $d \times d$ matrix of $\verticesV \times \verticesV$ blocks, it must be\footnote{To expand on this point, first consider the diagonal matrix $Q \in O(d)$ with $(-1,1,\dots,1)$ on the diagonal. Conjugation by this $Q$ reverses the sign of all the blocks in the first column and first row except the block on the diagonal. Thus all these blocks (and by extension, all the off-diagonal blocks) must be zero matrices. Now consider the permutation matrix $Q \in O(d)$ which swaps $i$ and $j$. Conjugation by this matrix swaps the $i$-th and $j$-th diagonal blocks (and the $ij$ and $ji$ off-diagonal blocks) so the diagonal blocks must be equal.} blockwise a $d \times d$ scalar matrix, as scalar matrices are the only $d \times d$ matrices fixed by $O(d)$. 

Since $\mean{\nu}{XX^T}$ is an average of symmetric positive semidefinite matrices $XX^T$, it is symmetric and positive semidefinite. Therefore, the $\verticesV \times \verticesV$ diagonal block $\Sigma_\verticesV$ is symmetric and positive semidefinite as well. 
\end{proof}

\begin{proposition}
\label{prop:covariance structure for ED}
If $\mu$ is an $O(d)$-invariant probability measure on $\ED$, then the $d\edgesE \times d\edgesE$ matrix $\mean{\mu}{W W^T} = I_d \otimes \Sigma_\edgesE$ where $\Sigma_\edgesE$ is a symmetric positive semidefinite $\edgesE \times \edgesE$ matrix giving the expectations of products of \emph{a fixed coordinate} of the displacements assigned to edges in $\graphG$. 

As above, if we view $W$ as a $d \times \edgesE$ matrix, $\mean{\mu}{W W^T} = (\tr \Sigma_\edgesE) I_d$ and $\mean{\mu}{W^T W} = d \, \Sigma_\edgesE$.
\label{prop:block structure of WWt}
\end{proposition}

The proof is the same as that of the previous proposition. We have now shown that different coordinates of the vertex positions and edge vectors~\emph{must always be uncorrelated}, not only when they are assigned to different vertices or edges, but even when they are different coordinates of the same vertex position or edge vector! In phantom network theory one can go much further-- because of the special properties of Gaussian probability measures, the different coordinates are actually \emph{independent}, and not just uncorrelated. However, it is quite surprising to find the same structure for the covariance matrix even in cases (like the FENE potential or freely-jointed chain) where it is very clear that the different coordinates of an edge vector are dependent random variables. We note that so far, we have not used concentration of $\mug$ and $\nug$ on their respective subspaces or the fact that $\nug$ is the pushforward of $\mug$; only $O(d)$-invariance. We now introduce these other hypotheses, and see that they imply that with our inner products, \emph{edge covariances and vertex covariances are exactly the same.}
\begin{proposition}\label{prop:matching covariances}
Suppose that either
\begin{enumerate}
\item $\mug$ is any $O(d)$-invariant probability measure on $\EDprod$ which is concentrated on $\im \bdy^*$ and $\nug = \bdy^{*+}_\sharp \mug$ or, 
\item $\nug$ is any $O(d)$-invariant probability measure on $\VPprod$ which is concentrated on the centered configurations $\im \bdy^{*+}$ and $\mug = \bdy^*_\sharp \nug$. 
\end{enumerate}
These hypotheses are equivalent. Further, if $X, Y \in \VP$ are centered, then $\cov{\nug}{X}{Y} = \cov{\mug}{\bdy^*X}{\bdy^*Y}$. If $W, U \in \ED$ are in $\im \bdy^*$, then $\cov{\mug}{W}{U} = \cov{\nug}{\bdy^{*+}W}{\bdy^{*+}U}$.

If $X \in \VP$ is in the orthogonal complement $\ker \bdy^*$ of the centered configurations and $Y \in \VP$, then $\cov{\nug}{X}{Y} = 0$. If $W \in \ED$ is in $(\im \bdy^*)^\perp$ and $U \in \ED$, then 
$\cov{\mug}{W}{U} = 0$.
\end{proposition}

\begin{proof}
Suppose $X, Y$ are centered in $\VP$. We begin by computing
\begin{align*}
\cov{\mug}{\bdy^* X}{\bdy^* Y} 
= \mean{\mug}{\EDp{W}{\bdy^* X} \EDp{W}{\bdy^* Y}}
& = \int_{W \in \ED} \EDp{W}{\bdy^* X} \EDp{W}{\bdy^*Y} \mug(dW) \\
& = \int_{W \in \im \bdy^*} \EDp{W}{\bdy^* X} \EDp{W}{\bdy^*Y} \mug(dW)
\end{align*}
since $\mug$ is concentrated on $\im \bdy^*$. Therefore, since $\mug = (\bdy^*)_\sharp \nug$, this integral is equal to
\begin{align*}
\cov{\mug}{\bdy^* X}{\bdy^* Y} 
&= \int_{Z \in \VP} \EDp{\bdy^* Z}{\bdy^* X} \EDp{\bdy^*Z}{\bdy^*Y} \nug(dZ) \\
&= \int_{Z \in (\ker \bdy^*)^\perp} \EDp{\bdy^* Z}{\bdy^* X} \EDp{\bdy^*Z}{\bdy^*Y} \nug(dZ) \\
&= \int_{Z \in (\ker \bdy^*)^\perp} \VPp{Z}{X} \VPp{Z}{Y} \nug(dZ) = \cov{\nug}{X}{Y}.
\end{align*}
Here we have used the fact that $\nug$ is supported on the centered configurations $(\ker \bdy^*)^\perp$ in the second equality, and~\prop{bdystar is a partial isometry} in the third. 

The second half of the proof follows from the fact that $\nug$ and $\mug$ are concentrated on $(\ker \bdy^*)^\perp$ and $\im \bdy^*$. In general, if we have a random vector $x$ chosen according to a probability measure $\rho$ which is concentrated on a subspace $S$ of $\Vprod$, then for any $u \in S^\perp$, 
\begin{equation*}
\cov{\rho}{u}{v} = \int_{x \in V} \Vp{u}{x} \Vp{v}{x} \rho(dx) = \int_{x \in S} \Vp{u}{x} \Vp{v}{x} \rho(dx) = 0,
\end{equation*}
because we know that $\Vp{u}{x} = 0$ since $x \in S$ and $u \in S^\perp$.
\end{proof}

\begin{corollary}\label{cor:XXt and WWt}
With the hypotheses of~\prop{matching covariances}, and the notation of~\prop{block structure of XXt} and~\prop{block structure of WWt}, we have $\bdy^* \mean{\nug}{XX^T} \bdy^{*T} = \mean{\mug}{WW^T}$ and so $\bdy^T \Sigma_\verticesV \bdy = \Sigma_\edgesE$. Further, 
\begin{equation*}
\mean{\nug}{XX^T} = \bdy^{*+} \mean{\mug}{WW^T} \bdy^{*T+} \quad\text{and}\quad \Sigma_\verticesV = \bdy^{T+} \Sigma_\edgesE \bdy^+.
\end{equation*}
\end{corollary}

\begin{proof}
Suppose $U, V \in \VP$ are centered. Then $\cov{\nug}{U}{V} = \cov{\mug}{\bdy^* U}{\bdy^* V}$. Using~\eqn{covariance in general inner product},
\begin{equation*}
\cov{\nug}{U}{V} 
= \mean{\nug}{\VPp{X}{U} \VPp{X}{V}} 
= \dotp{\tilde{L}^* U}{\mean{\nug}{XX^T} \tilde{L}^* V}.
\end{equation*}
Using~\prop{centered} and $\tilde{L}^* = L^* + (\frac{1}{\verticesV} \onesVV)^*$, we have
$\tilde{L}^*U = L^* U$ and $\tilde{L}^* V = L^* V$. So
\begin{equation}
\begin{aligned}
\cov{\nug}{U}{V} 
&= \dotp{L^* U}{\mean{\nu}{XX^T}L^* V} \\
&= \dotp{\bdy^{T*} \bdy^* U}{\mean{\nug}{XX^T}\bdy^{*T} \bdy^* V} \\
&= \dotp{\bdy^* U}{\bdy^* \mean{\nug}{XX^T} \bdy^{*T} \bdy^*V}
\end{aligned}
\label{eq:covariance in nu of U and V}
\end{equation}
On the other hand, using~\eqn{covariance in general inner product} and the fact that the dot product on $\EDprod$ is standard, 
\begin{equation}
\cov{\mug}{\bdy^* U}{\bdy^* V} = \dotp{\bdy^* U}{\mean{\mug}{WW^T} \bdy^* V}
\label{eq:covariance in mu of bdystar U and bdystar V}
\end{equation}
We now have two matrices $\bdy^* \mean{\nu}{XX^T} \bdy^{*T}$ and $\mean{\mu}{WW^T}$. The kernel of each matrix contains the loop space $\ker \bdy^{*T} = (\im \bdy^*)^\perp$ in $\ED$. The image of each matrix is contained in $\im \bdy^*$. Combining~\eqn{covariance in nu of U and V} and~\eqn{covariance in mu of bdystar U and bdystar V} and using the shared kernel, we have for all $Z \in \im \bdy^*$ and $P \in \ED$ that
\begin{equation*}
\dotp{Z}{\mean{\mug}{WW^T} P} = \dotp{Z}{\bdy^* \mean{\nug}{XX^T} \bdy^{*T} P}.
\end{equation*}
Thus we can conclude that $\mean{\mug}{WW^T} P = \bdy^* \mean{\nug}{XX^T} \bdy^{*T} P$ for all $P \in \ED$ and hence that the two matrices are equal. This proves the first part. 

Multiplying this identity on the left by $\bdy^{*+}$ and on the right by $\bdy^{*T+}$, 
\begin{equation*}
\bdy^{*+} \bdy^* \mean{\nug}{XX^T} \bdy^{*T} \bdy^{*T+} = \bdy^{*+} \mean{\mug}{WW^T} \bdy^{*T+}.
\end{equation*}
But $\bdy^{*+} \bdy^* = \proj_{\im \bdy^{*+}} = \proj_{(\ker \bdy^*)^\perp}$ is orthogonal projection onto the centered configurations. Since $\nu$ is already supported on the centered configurations, $\im \mean{\nug}{XX^T} \subset (\ker \bdy^*)^\perp$ and composing $\mean{\nug}{XX^T}$ with this matrix has no effect. 

Similarly, $\bdy^{*T} \bdy^{*T+} = \proj_{\im \bdy^{*T}} = \proj_{\im \bdy^{*+}} = \proj_{(\ker \bdy^*)^\perp}$ is the same projection. Again, $\ker \mean{\nug}{XX^T} \supset \ker \bdy^*$, so projecting to the orthogonal complement of this kernel before applying $\mean{\nug}{XX^T}$ has no effect. This proves that $\mean{\nug}{XX^T} = \bdy^{*+} \mean{\mug}{WW^T} \bdy^{*T+}$, as desired.
\end{proof}

\section{Expected radius of gyration; phantom network theory}

We have now proved a quite general structural result about variances. In practice, we are often interested in computing the expectation of the radius of gyration\footnote{This expectation is usually referred to as ``the radius of gyration'', even though it's technically an ensemble average of the radius of gyration of all possible conformations of the network.} of a polymer structure. 
\begin{definition}\label{def:gyradius}
The \emph{expected radius of gyration} of a polymer whose vertex positions $X \in \VP$ are distributed according to probability measure $\nug$ concentrated on centered configurations in $\VP$ is 
\begin{equation}\label{eq:gyradius}
\gyradius_{\nug} = \frac{1}{\verticesV} \mean{\nug}{\sum \norm{X_{ij}}^2} = \frac{1}{\verticesV} \tr \mean{\nug}{XX^T}.
\end{equation}
\end{definition}
We now give a general formula for the expected radius of gyration:
\begin{theorem}\label{thm:gyradius formula}
If $\nug = (\bdy^{*+})_\sharp \mug$ where $\mug$ is any $O(d)$-invariant probability measure concentrated on $\im \bdy^*$, then if $\Sigma_\verticesV$ and $\Sigma_\edgesE$ are as defined in~\prop{block structure of XXt} and~\prop{block structure of WWt}, we have
\begin{align*}
\gyradius_{\nug} = \frac{d}{\verticesV} \tr (\Sigma_\verticesV)
	      = \frac{d}{\verticesV} \tr (\bdy^{T+} \Sigma_\edgesE \bdy^{+}).
\end{align*}
\end{theorem}

\begin{proof}
We only have to combine~\eqn{gyradius} with~\prop{block structure of XXt} to see the first line. The second line follows directly from~\cor{XXt and WWt}.
\end{proof}

Here is an easy corollary. 
\begin{corollary}\label{cor:phantom network theory}
In James--Guth phantom network theory $\mu$ is a standard Gaussian on $\ED$. For any graph $\graphG$, $\mu$ is admissible and compatible with $\graphG$ and $\mug = \mugz$ is a standard Gaussian on $\im \bdy^*$. Further, 
\begin{equation}
\gyradius_{\nug} =  \frac{d}{\verticesV} \tr L^+.
\end{equation}
\end{corollary}

\begin{proof}
We showed in~\cor{phantom network theory works} that $\mu$ is admissible and compatible with $\graphG$, so $\mug$ is well-defined. Further, by~\prop{disintegration}, $\mug$ has $O(d)$-invariance and concentration on $\im \bdy^*$ so~\thm{gyradius formula} applies. It is easy to see that $\mean{\mug}{WW^T}$ is orthogonal projection onto $\im \bdy^*$. It follows from~\prop{block structure of WWt} that $\Sigma_\edgesE$ is the orthogonal projection $\bdy^T \bdy^{T+}$ to $\im \bdy^T$. Thus
\begin{equation*}
\gyradius_{\nug} = \frac{d}{\verticesV} \tr \bdy^{T+} (\bdy^{T} \bdy^{T+}) \bdy^+ = \frac{d}{\verticesV} \tr \bdy^{T+} \bdy^+ = \frac{d}{\verticesV} \tr L^+,
\end{equation*}
as claimed. \end{proof}
Here is another example of our method. 
\begin{proposition}\label{prop:cycle graph gyradius}
Suppose that $\graphG$ is the $n$-edge cycle graph, $\mu$ is the joint distribution of $n$ i.i.d.\,random edges in $\R^d$ chosen from some $O(d)$-invariant probability measure on $\R^d$, and $\mu$ is compatible with $\graphG$. Then $\mug$ is permutation invariant, all $\mean{\mug}{\norm{W(\edge_i)}^2}$ are equal and if their common value is $\lambda$ then 
\begin{equation*}
\gyradius_{\nug} = \frac{\lambda (\edgesE + 1)}{12}.
\end{equation*}
\end{proposition}

\begin{proof}
It is clear that $\mu$ is an $O(d)$-invariant probability distribution as it's the joint distribution of $O(d)$-invariant random variables. Thus $\mu$ is admissible. Further, $\mu$ is permutation invariant on edges since the individual distributions are independent and identical. 

We have assumed that $\mu$ is compatible with $\graphG$, so $\mug$ is well-defined. Permutation invariance of the conditional distribution $\mug$ can be proved by uniqueness of decompositions as we did in the proof of $O(d)$-invariance in~\prop{disintegration}. 

Using this permutation symmetry, all of the off-diagonal elements of the $\edgesE \times \edgesE$ matrix $\mean{\mug}{W^T W}$ of edge covariances are equal, and all of the on-diagonal elements are equal as well. Since the loop space of the cycle graph $\ker \bdy \in \EC$ is one-dimensional and spanned by $\mat{1}_{\edgesE \times 1}$, we know $\mat{1}_{\edgesE \times 1}$ spans $\ker \mean{\mug}{W^T W}$ and the sum of each row and column is zero. Thus
\begin{equation}
\mean{\mug}{W^T W} = \frac{\lambda \edgesE}{\edgesE - 1} I_\edgesE - \frac{\lambda}{\edgesE-1} \onesEE.
\end{equation}
where the diagonal element $\lambda$ is the expected squared norm of a single edge~\emph{in $\mug$}. 

We have $\verticesV = \edgesE$ and $\bdy$ is a (square) circulant matrix with first row $-1, 0, \dotsc, 0, 1$. It follows that $\bdy^T \bdy$ is a symmetric circulant matrix; its first row is $2, -1, \dots, 0, -1$. Since the row and column sums of $\bdy^T \bdy$ vanish, the row and column sums of $(\bdy^T \bdy)^+$ vanish as well and $\onesEE (\bdy^T \bdy)^+ = 0$. Further, the trace of $\bdy^T \bdy$ is\footnote{Eigenvalues of a circulant are well-known; to sum them use $\sum_{j=1}^{v-1} \csc^2 (\pi j)=(1/3)(v^2-1)$~\cite[4.4.6.5, p.\ 644]{Prudnikov:1986vp}.} $\frac{1}{12} (\edgesE^2-1)$. 

Now we are ready to compute. Since $\mug$ is (by construction) $O(d)$-invariant and concentrated on $\im \bdy^*$,~\thm{gyradius formula} applies. Rearranging it and using $d \Sigma_\edgesE = \mean{\mug}{W^T W}$,
\begin{align*}
\gyradius = \frac{1}{\edgesE} \tr \left(\mean{\mug}{W^T W} (\bdy^T \bdy)^+ \right)
 	      = \frac{1}{\edgesE} \tr \left(\lambda \frac{\edgesE}{\edgesE-1} (\bdy^T \bdy)^+ \right) 
	      = \frac{\lambda (\edgesE + 1)}{12}.
\end{align*}
\end{proof}

We note that we can satisfy the hypothesis that $\mu$ is compatible with $\graphG$ in many cases using decay estimates (cf. \prop{disintegration with density} and \cor{compatibility for joint distributions}). For the freely jointed ring polymer with $\edgesE$ equilateral edges, we can immediately recover the standard formula for $\gyradius$ (\cite{Zirbel2012}):
\begin{corollary}
\label{cor:freely jointed ring}
Let $\graphG$ be the $\edgesE$-edge cycle graph and $\mu$ be the product of (uniform) area measures on the product of spheres $(S^2)^\edgesE \subset \ED = (\R^3)^\edgesE$. $\mu$ is admissible and compatible with $\graphG$ and $\gyradius_{\nug}  = (\edgesE + 1)/12$.
\end{corollary}
\begin{proof}
Admissibility and compatibility were established above in~\cor{freely jointed ring admissible}, so there is a well-defined conditional probability $\mug = \mugz$ with $\lambda = \mean{\mug}{\norm{W(\edge_i)}^2} = 1$. Applying~\prop{cycle graph gyradius} completes the proof.
\end{proof}
We can also recover the standard result for the Gaussian ring polymer:
\begin{corollary}\label{cor:Gaussian ring}
Let $\graphG$ be the $\edgesE$-edge cycle graph and $\mu$ be the standard Gaussian on $\ED = (\R^3)^\edgesE$. $\mu$ is admissible and compatible with $\graphG$, $\mug = \mugz$ is well-defined, and 
\begin{equation*}
\gyradius_{\nug} = (\edgesE^2 - 1)/(12 \edgesE)
\end{equation*}
\end{corollary}

\begin{proof} 
In~\cor{phantom network theory}, we showed that this $\mu$ is compatible with any $\graphG$, along with the fact that $\mug = \mugz$ was a standard Gaussian on $\im \bdy^*$. The computation of $\lambda = (\edgesE-1)/\edgesE$ can be done very easily by the method of~\prop{projections} below.
\end{proof}

\section{Chain maps; pushforwards to simpler graphs}
\label{sec:chain maps}

We now want to consider the distribution of more local quantities in a graph-- for instance, the squared distance between a particular pair of vertices instead of the ensemble sum which appears in the radius of gyration. To do this, it's helpful to compute the marginal distribution of the subset of vertices which are needed to compute the local quantity in question and then take expectations with respect to this marginal distribution. We start by defining the random variables we'll consider. 
\begin{definition} \label{def:expressed}
Suppose we have graphs $\graphG$ and $\graphG'$ and an injective map $f_0 \co \VC' \rightarrow \VC$.
A function $g \co \VP \rightarrow \R$ is~\emph{expressed in terms of $\graphG'$} if there is a map $g' \co \VP' \rightarrow \R$ so that $g = g' \circ f_0^*$.
\end{definition}
It is standard that
\begin{proposition} \label{prop:expectations under pushforwards}
If $g$ is expressed in terms of $\graphG'$ and $\nug$ is any probability distribution on $\VP$, then $\mean{\nug}{g} = \mean{(f_0^*)_\sharp \nug}{g'}$.
\end{proposition}
This proposition is evidently true but it's not very useful in practice, as computing an expectation with respect to $(f_0^*)_\sharp \nug$ is likely just as hard as computing one with respect to $\nug$ in the first place. On the other hand, if $\nug$ is a probability distribution concentrated on the centered configurations which we have obtained from a probability distribution $\mu$ on $\ED$, it would be useful to construct some probability distribution $\mu'$ on $\ED'$ so that the corresponding $\nugp$ had $\mean{\nug}{g} = \mean{\nugp}{g'}$. The purpose of this section is to establish a construction of $\mu'$ which, coupled with mild restrictions on $g'$, will accomplish this goal. We start by connecting the map $f_0$ on vertex chains to a corresponding map $f_1$ between edge chains.
\begin{definition}
Given two graphs $\graphG$ and $\graphG'$ and a pair of maps $f_0 \co \VC' \rightarrow \VC$ and $f_1 \co \EC' \rightarrow \EC$, we say that $f_0$ and $f_1$ are~\emph{chain maps} if $\bdy f_1 = f_0 \bdy'$.
\end{definition}
Chain maps $f_0$ and $f_1$ induce maps $f_0^* \co \VP \rightarrow \VP'$ and $f_1^* \co \ED \rightarrow \ED'$ and the definition ensures the squares below commute: 
\begin{equation} \label{eq:diagrams}
\begin{tikzcd}
\VC  & \arrow[l,"\bdy"]  \EC \\
\VC'   \arrow[u,"f_0"] & \EC'  \arrow[l,"\bdy'"] \arrow[u,"f_1"] 
\end{tikzcd} 
\qquad 
\begin{tikzcd}
\VP  \arrow[r,"\bdy^*"] \arrow[d,"f_0^*"] &   \ED \arrow[d,"f_1^*"] \\
\VP' \arrow[r,"(\bdy')^*"]  & \ED'    
\end{tikzcd}
\end{equation}
We will think of $\graphG'$ as a simpler graph which we include in $\graphG$ by the chain maps $f_0$ and $f_1$. The chain map hypothesis ensures that our assignments of edges and vertices are compatible with each other and the graph structures.

\begin{proposition}\label{prop:chain map pushforwards}
Suppose $f_0$ and $f_1$ are injective chain maps $\graphG' \rightarrow \graphG$, and $\graphG$ and $\graphG'$ have the same cycle rank $\edgesE - \verticesV + 1 = \edgesE' - \verticesV' + 1$. 
Then 
\begin{enumerate}
\item $\proj_{\ID'} f_1^*$ is an invertible linear map between $\ID$ and $\ID'$, \label{projidprime f1star is invertible}
\item \label{mu and mup compatible with gprime} $\mu$ is compatible with $\graphG$ $\implies$ $\mup := \pushmu$ is compatible with $\graphG'$, 
\item \label{mugwp are pushmugw} Since $\mu$ is compatible with $\graphG$, by definition there is an open ball $U \subset \ID$ centered at $0$ so that $\mugw$ is defined for all $W \in U$. If we let $W' := (\proj_{\ID'} f_1^*) W$ then for all $W'$ in an open ball $U'$ centered at $0$ in $\ID'$, there are corresponding $W \in U$ so that 
\begin{equation*}
\mugwp = \pushmugw.
\end{equation*}
\end{enumerate}
\end{proposition}
Here the $\mugw$ and $\mugwp$ are the conditional probabilities guaranteed by the $\graphG$-compatibility of $\mu$ and the $\graphG'$-compatibility of $\mup$ (cf.\,\defn{G-compatible}). 

\begin{proof}
We start by working out some properties of $f_1$ and $f_1^*$ which follow from our hypotheses.

\begin{claim}\label{claim:chain maps}
We have $f_1 (\ker \bdy') = \ker \bdy$. Thus $\ker f_1^* \subset \im \bdy^*$.
\end{claim}

\begin{proof}
If $w' \in \ker \bdy'$, then $0 = f_0 \bdy' w' = \bdy f_1 w'$, so $f_1 w' \in \ker \bdy$. Thus $f_1(\ker \bdy') \subset \ker \bdy$. Since $f_1$ is injective, $\dim f_1 (\ker \bdy') = \dim \ker \bdy'$. But $\dim \ker \bdy' = \dim \ker \bdy$ because each is the cycle rank $\edgesE - \verticesV + 1 = \edgesE' - \verticesV' + 1$. Thus $f_1(\ker \bdy') = \ker \bdy$.

Next, since $\ker \bdy \subset \im f_1$, we know $(\im f_1)^0 \subset (\ker \bdy)^0$. But $\ker f_1^* = (\im f_1)^0$ and $(\ker \bdy)^0 = \im \bdy^*$, so this proves $\ker f_1^* \subset \im \bdy^*$.
\end{proof}

\begin{claim} \label{claim:images}
We have $f_1^*(\im \bdy^*) = \im (\bdy')^*$ and $(f_1^*)^{-1}(\im (\bdy')^*) = \im \bdy^*$.
\end{claim}

\begin{proof}
Since $f_1(\ker \bdy') = \ker \bdy$, $f_1^* (\ker \bdy)^0 = (\ker \bdy')^0$, or $f_1^*(\im \bdy^*) = \im (\bdy')^*$. It follows immediately that $(f_1^*)^{-1}(\im (\bdy')^*) = (f_1^*)^{-1}(f_1^*(\im \bdy^*))$, so $\im \bdy^* \subset (f_1^*)^{-1}(\im (\bdy')^*)$. 

Now suppose $f_1^* W \in \im (\bdy')^*$. Then $f_1^* W = (\bdy')^* X' = (\bdy')^* f_0^* X = f_1^* \bdy^* X$ for some $X \in \VP$, where the second equality follows from surjectivity of $f_0^*$ (\cor{injective dual to surjective}). This means that $\bdy^*X - W \in \ker f_1^*$. But $\ker f_1^* \subset \im \bdy^*$ by the last claim, so $\bdy^*X - W \in \im \bdy^*$, and $W \in \im \bdy^*$, as required.
\end{proof}

\begin{claim} \label{claim:id isomorphism}
The map $\proj_{\ID'} f_1^*$ is a linear isomorphism from $\ID$ to $\ID'$. 
\end{claim}

\begin{proof}
First, it's clear that $\im \proj_{\ID'} f_1^* \subset \ID'$. Next, we show that $\proj_{\ID'} f_1^*$ restricted to $\ID$ has kernel $0$. Suppose that $W \in \ID$ has $\proj_{\ID'} f_1^* W = 0$. Then $f_1^* W \in \ker \proj_{\ID'} = \im (\bdy')^*$. By~\clm{images}, this implies $W \in \im \bdy^*$. But $\im \bdy^* \cap \ID = 0$, so this means $W = 0$. Thus $\proj_{\ID'} f_1^*$ restricted to $\ID$ is injective. But $\dim \ID' = d \xi(\graphG') = d \xi(\graphG) = \dim \ID$, which means that $\proj_{\ID'} f_1^*$ restricted to $\ID$ is an isomorphism.
\end{proof}

\begin{claim}\label{claim:proj f1 proj is proj f1}
We claim that $\proj_{\ID'} f_1^* \proj_{\ID} = \proj_{\ID'} f_1^*$.
\end{claim}

\begin{proof}
Any $W \in \ED$ can be written uniquely as $W = \bdy^* X + W_{\ID}$ for some $X \in \VP$ and $W_{\ID} \in \ID$. Now $f_1^* \bdy^* X \in \im (\bdy')^* = \ker \proj_{\ID'}$ by~\clm{images}. Thus $\proj_{\ID'} f_1^* W = \proj_{\ID'} f_1^* W_{\ID}$. But $\proj_{\ID} W = W_{\ID}$, so 
$\proj_{\ID'} f_1^* W_{\ID} = \proj_{\ID'} f_1^* \proj_{\ID} W$, as required.
\end{proof}

\begin{claim}\label{claim:fiber is pushforward of fiber}
Recalling that $W' = \proj_{\ID'} f_1^* W$, we claim that $\proj_{\ID'}^{-1}(W') = f_1^* \proj_{\ID}^{-1}(W)$.
\end{claim}

\begin{proof}[Proof of claim]
We first show that $f_1^* \proj_{\ID}^{-1}(W) \subset \proj_{\ID'}^{-1}(W')$. Suppose $Z \in \proj_{\ID}^{-1}(W)$. Then $\proj_{\ID} Z = W$, so $\proj_{\ID'} f_1^* \proj_{\ID} Z = \proj_{\ID'} f_1^* W = W'$. But by~\clm{proj f1 proj is proj f1}, $\proj_{\ID'} f_1^* \proj_{\ID} Z = \proj_{\ID'} f_1^* Z$. Therefore, $f_1^* Z \in \proj_{\ID'}^{-1}(W')$, as required.

Now we show that $\proj_{\ID'}^{-1}(W') \subset f_1^* \proj_{\ID}^{-1}(W)$. Suppose that we have $Z' \in \proj_{\ID'}^{-1}(W')$. Then $\proj_{\ID'} Z' = W'$. Now $f_1^*$ is surjective, so there exists some $Z \in \ED$ with $f_1^* Z = Z'$. We now know that $\proj_{\ID'} f_1^* Z = W'$. We must show that $Z \in \proj_{\ID}^{-1}(W)$. By~\clm{proj f1 proj is proj f1}, we know 
\begin{equation*}
\proj_{\ID'} f_1^* \proj_{\ID} Z = \proj_{\ID'} f_1^* Z = W' = \proj_{\ID'} f_1^* W.
\end{equation*}
But by~\clm{id isomorphism}, since $\proj_{\ID} Z \in \ID$ and $W \in \ID$, this implies that $\proj_{\ID} Z = W$, as required.
\end{proof}

\begin{claim}
$(f_1^*)_\sharp \mu$ is an admissible measure on $\ED$.
\end{claim}

\begin{proof}
Since $f_1^*$ is a linear map, it is Lipschitz. Further, Lipschitz pushforwards of Radon probability measures with finite first moment are also Radon probability measures of finite first moment~\cite{Fritz2019}, so $(f_1^*)_\sharp \mu$ is admissible.
\end{proof}

We are now ready for the body of the proof. We noted in~\defn{tjur conditional probability} that Tjur conditional probabilities are unique if they are defined. Therefore, we can establish both~\ref{mu and mup compatible with gprime} and~\ref{mugwp are pushmugw} by showing that the $\pushmugw$ are conditional probabilities for $\mu'$ given $\proj_{\ID'} Z' = W'$.

\begin{claim}
The $\mugw$ are conditional probabilities for $\mu$ given $\proj_{\ID'} f_1^* \proj_{\ID}(Z) = W'$.
\end{claim}

\begin{proof}[Proof of claim]
By~\clm{proj f1 proj is proj f1}, we know that $\mu'_{\ID'} = (\proj_{\ID'})_\sharp \mu' = (\proj_{\ID'} f_1^* \proj_{\ID})_\sharp \mu$.

So suppose we fix some $f \in \mathcal{K}(\ED)$ and some $\epsilon > 0$. By hypothesis, there is some open neighborhood $V$ of $W$ in $\ID$ so that for every $B \subset V$ with $\mu_{\ID}(B) > 0$ we have $\abs{\mugw(f) - \mu^B(f)} < \epsilon$. Now the map $\proj_{\ID'} f_1^*$ is a linear isomorphism from $\ID$ to $\ID'$ by~\clm{id isomorphism}. Therefore, there is some open neighborhood $V'$ of $W' = \proj_{\ID'} f_1^* W$ so that $(\proj_{\ID'} f_1^*)^{-1}(V') \subset V$. Now suppose we have any $B' \subset V'$ with $\mu'_{\ID'}(B') > 0$. Defining $B$ to be the inverse image $B := (\proj_{\ID'} f_1^*)^{-1}(B')$, and applying the definition of pushforward, we now know that
\begin{equation*}
0 < \mu'_{\ID'}(B') = (\proj_{\ID})_\sharp \mu( (\proj_{\ID'} f_1^*)^{-1}(B') ) = \mu_{\ID}(B).
\end{equation*}
Further, since $B' \subset V'$, we know $B \subset V$. We now compute
\begin{align*}
(\mu')^{B'}(f) &= \frac{1}{\mu'_{\ID'}(B')} \int_{(\proj_{\ID'} f_1^* \proj_{\ID})^{-1}(B')} f(Z) \mu(dZ) \\
&= \frac{1}{\mu_{\ID}(B)} \int_{(\proj_{\ID})^{-1}(B)} f(Z) \mu(dZ) = \mu^B(f)
\end{align*}
This proves that $\abs{\mugw(f) - (\mu')^{B'}(f)} = \abs{\mugw(f) - \mu^B(f)} < \epsilon$, and hence that the $\mugw$ are conditional probabilities for $\mu$ given $\proj_{\ID'} f_1^* \proj_{\ID}(W) = W'$.
\end{proof}

Applying~\thm{decomposition and existence}, we now know that the $\mugw$ and $\mu'_{\ID'}$ are a decomposition of $\mu$ with respect to the map $\proj_{\ID'} f_1^* \proj_{\ID} = \proj_{\ID'} f_1^*$.
By~\lem{decompositions push forward}, the $\pushmugw$ and $\mu'_{\ID'}$ are thus a decomposition of $f_1^* \mu = \mu'$ with respect to $\proj_{\ID'}$. One last application of~\thm{decomposition and existence} shows that the $\pushmugw$ are conditional probabilities for $\mu'$ given $\proj_{\ID'} (Z') = W'$, as desired.
\end{proof}

We can now prove the main theorem of the section.
\begin{theorem}\label{thm:chain maps and probability}
Suppose we have injective chain maps $f_0$ and $f_1$ between connected graphs $\graphG$ and $\graphG'$ of the same cycle rank $\xi(\graphG)$, together with a measure $\mu$ on $\ED$ which is compatible with $\graphG$ and its pushforward $\mup = \pushmu$ on $\ED'$. Then $\mup$ is compatible with $\graphG'$ and 
the corresponding $\nug := \bdy^{*+}_\sharp \mug$ and $\nugp := (\bdy')^{*+}_\sharp \mugp$ measures are related by
\begin{equation*}
(\proj_{\im (\bdy')^{*+}} f_0^*)_\sharp \nug = \nugp.
\end{equation*}
It follows that if $g$ is any $O(d)$ and translation-invariant function $g \co \VP \rightarrow \R$ which is expressed in terms of $\graphG'$, we have $\mean{\nug}{g} = \mean{\nugp}{g'}$.
\end{theorem}
We remind the reader that $\graphG$-compatibility of $\mu$ $\implies$ $\graphG'$-compatibility of $\mu'$ was established in~\prop{chain map pushforwards}. Compatibility guarantees the existence of the conditional probabilities $\mug, \nug := \bdy^{*+}_\sharp \mug$ and $\mugp, \nugp := (\bdy')^{*+}_\sharp \mugp$. The map $g'$ is defined so $g' \circ f_0^* = g$. Its existence is guaranteed by~\defn{expressed}, since we have assumed that $g$ is expressed in terms of $\graphG'$.

\begin{proof}
We are first going to do a bit of linear algebra.
\begin{claim} \label{claim:f0star and ker bdystar}
$f_0^*(\ker \bdy^*) = \ker (\bdy')^*$.
\end{claim}

\begin{proof}[Proof of claim] 
We know that $\bdy f_1 = f_0 \bdy'$ since $f_0$ and $f_1$ are chain maps. Thus if $x \in \im \bdy'$, then $f_0 x = f_0 \bdy' w = \bdy f_1 w$, so $f_0 x \in \im \bdy$. Thus $f_0( \im \bdy' ) \subset \im \bdy$. But this means that $(\im \bdy')^0 \subset f_0^* (\im \bdy)^0$, or $\ker (\bdy')^* \subset f_0^*(\ker \bdy^*)$. However, from~\prop{ker and im of bdystar}, we know that $\dim \ker (\bdy')^* = \dim \ker \bdy^* = d$. Since $\dim f_0^*(\ker \bdy^*) \leq \dim \ker \bdy^* = d$, we have shown that $\ker (\bdy')^* = f_0^*(\ker \bdy^*)$. 
\end{proof}

\begin{claim}
The map $g'$ is $O(d)$ and translation invariant.
\end{claim}

\begin{proof}[Proof of claim]
By definition, $g(X) = g'(f_0^*(X))$. To show that $g'$ is translation-invariant, we observe that the translations given by $\ker (\bdy')^*$. So if $Y'$ is a translation, $Y' = f_0^* Y$ (by the last~\clm{f0star and ker bdystar}). Further, $f_0^*$ is surjective, so for any $X' \in \VP'$, there is some $X \in \VP$ so that $X' = f_0^* X$. Now suppose $X \in \VP$ and $Y \in \ker (\bdy')^*$: 
\begin{equation*}
g'(X' + Y') = g'(f_0^*(X) + f_0^*(Y)) = g'(f_0^*(X + Y)) = g(X + Y) = g(X) = g'(f_0(X)) = g'(X').
\end{equation*}
where the step $g(X+Y) = g(X)$ follows from $Y \in \ker \bdy^*$ and the translation-invariance of $g$. 

By \lem{linear equivariance}, $f_0^* = I_d \otimes f_0^T$ is $GL(d)$-equivariant, and hence in particular $O(d)$-equivariant. To see what this means precisely, suppose that $Q \in O(d)$, so that the action of $Q$ on $\VP$ is $Q \otimes I_{\verticesV}$ and the action of $Q$ on $\VP'$ is $Q \otimes I_{\verticesV'}$. Since $f_0^T$ is a $\verticesV' \times \verticesV$ matrix: 
\begin{equation*}
(Q \otimes I_{\verticesV'}) f_0^* = (Q \otimes I_{\verticesV'})(I_d \otimes f_0^T) 
= (Q \otimes f_0^T) = (I_d \otimes f_0^T)(Q \otimes I_{\verticesV}) = f_0^* (Q \otimes I_{\verticesV}).
\end{equation*}
where we used \lem{mixed product} in the middle steps. We then have
\begin{equation*}
g'( (Q \otimes I_{\verticesV'}) X' ) = g'( (Q \otimes I_{\verticesV'}) f_0^* X ) =
g' (f_0^* (Q \otimes I_{\verticesV} X)) = g((Q \otimes I_{\verticesV}) X) =
g(X) = g'(f_0^* X) = g'(X').
\end{equation*}
which completes the proof.
\end{proof}
We know from~\prop{expectations under pushforwards} that $\mean{\nug}{g} = \mean{(f_0^*)_\sharp \nug}{g'}$. So we must show that $\mean{(f_0^*)_\sharp \nug}{g'} = \mean{\nugp}{g'}$. Now we have just proved that $g'$ is translation-invariant, so the expectation of $g'$ with respect to $(f_0^*)_\sharp \nug$ is equal to the expectation of $g'$ with respect to $(\proj_{\im (\bdy')^{*+}})_\sharp (f_0^*)_\sharp \nug$. It now suffices to show that $(\proj_{\im (\bdy')^{*+}} f_0^*)_\sharp \nug = \nugp$.

After all our work above, this is mostly a matter of unpacking definitions. Recall that $\nug := \bdy^{*+}_\sharp \mug$ and $\nugp := (\bdy')^{*+}_\sharp \mugp$. In the proof of~\prop{chain map pushforwards}, we showed $\mugp = \pushmug$. Now we can just compute:
\begin{align*}
(\proj_{\im (\bdy')^{*+}} f_0^*)_\sharp \nug &= 
(\proj_{\im (\bdy')^{*+}} f_0^* \bdy^{*+})_\sharp \mug 
= ((\bdy')^{*+} (\bdy')^* f_0^* \bdy^{*+})_\sharp \mug \\
&= ((\bdy')^{*+} f_1^* \bdy^* \bdy^{*+})_\sharp \mug 
= ((\bdy')^{*+} f_1^*)_\sharp (\proj_{\im \bdy^*})_\sharp \mug \\
&= ((\bdy')^{*+} f_1^*)_\sharp \mug = (\bdy')^{*+}_\sharp \pushmug \\
&= (\bdy')^{*+}_\sharp \mugp = \nugp,
\end{align*}
where $(\proj_{\im \bdy^*})_\sharp \mug = \mug$ because $\mug$ is concentrated on $\im \bdy^*$ by construction. 
\end{proof}

\section{Applications}

Proving~\thm{chain maps and probability} was somewhat complicated, but applying the theorem is a much easier process. We now give several examples which show how this result can greatly simplify calculations and numerical experiments regarding network polymers. We start by analyzing in some detail a very common model: subdivided graphs.

\begin{definition}\label{def:subdivision setup}
The $n$-part edge subdivision $\graphG_n$ of a multigraph $\graphG$ is the graph obtained by dividing each edge of $\graphG$ into $n$ smaller edges (see \figr{subdivisions}) oriented to agree with the original graph.

If $\graphG$ has $\verticesV$ vertices $\vertex_1, \dotsc, \vertex_{\verticesV}$  then $\graphG_n$ has $\verticesV$~\emph{junction vertices} $\vertex_{1 0}, \dotsc, \vertex_{\verticesV 0}$ corresponding to the vertices of $\graphG$ and $(n-1) \edgesE$ \emph{subdivision vertices} $\vertex_{11}, \dotsc, \vertex_{1(n-1)}, \vertex_{21}, \dotsc, \vertex_{\edgesE (n-1)}$ located along the subdivided edges.

If $\graphG$ has $\edgesE$ edges $\edge_1, \dotsc, \edge_{\edgesE}$, then $\graphG_n$ has $n\edgesE$ edges $\edge_{11}, \dotsc, \edge_{1n}, \edge_{21}, \dotsc, \edge_{\edgesE n}$. We will call each group $\edge_{j1}, \dotsc, \edge_{jn}$ the~\emph{subdivided edge} corresponding to $\edge_j$ in $\graphG$. 

There are canonical chain maps $f_0(\vertex_i) = \vertex_{i0}$ and $f_1(\edge_j) = \edge_{j1} + \cdots + \edge_{jn}$ from $\graphG$ to $\graphG_n$ which take vertices of $\graphG$ to the corresponding junction vertices in $\graphG_n$ and edges of $\graphG$ to the corresponding subdivided edges in $\graphG_n$.

We will reserve our usual notations $\mu, \mug, \bdy, \EC, \VC, \VP, \ED$ to refer to $\graphG$ and use the notations $\mu_n, \mu_{\graphG_n}, \bdy_n, \EC_n, \VC_n, \VP_n, \ED_n$ for the corresponding objects for the subdivided graph $\graphG_n$.
\end{definition}

\begin{figure}
\hphantom{.}
\hfill
\includegraphics[width=2in]{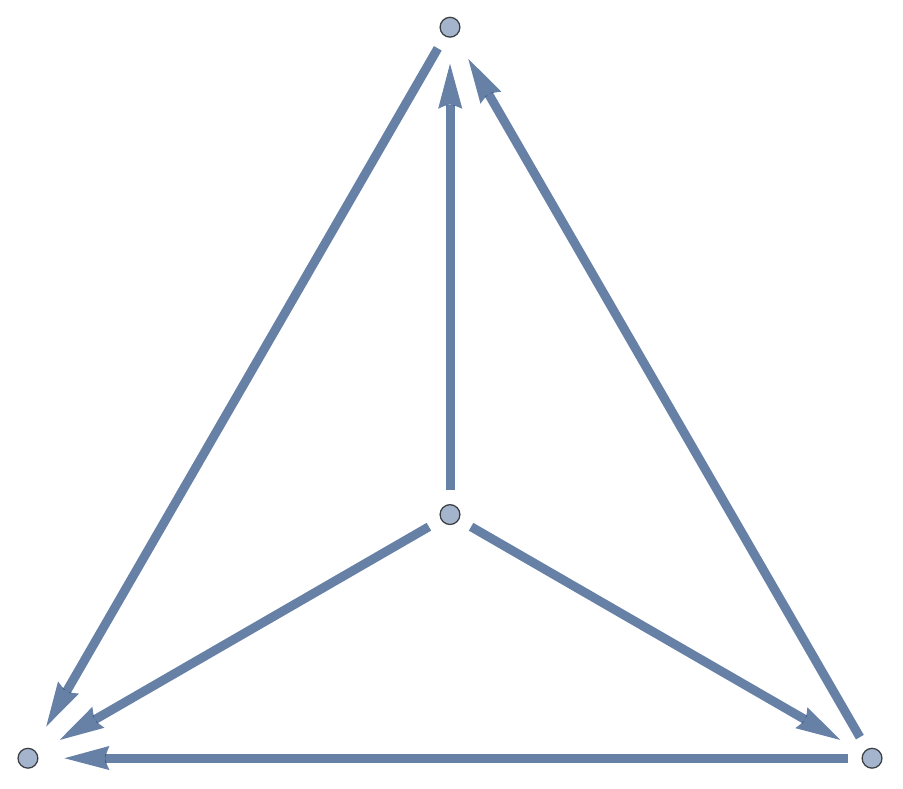}
\hfill
\includegraphics[width=2in]{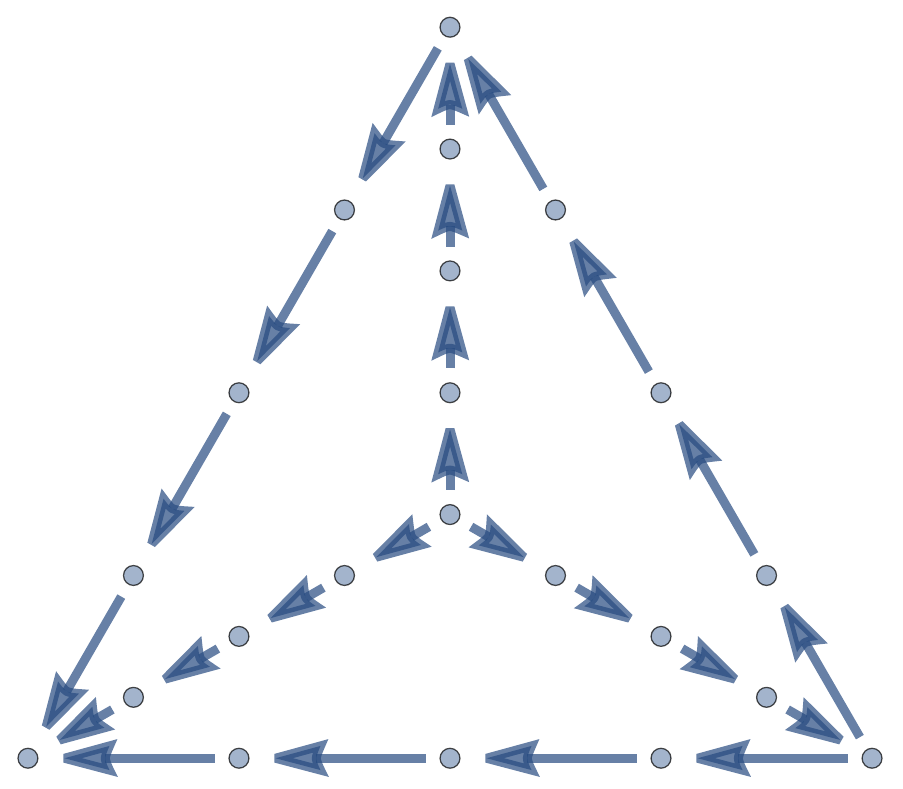}
\hfill
\hphantom{.}
\caption{A directed $\alpha$-graph $\graphG$ (left) and its four-part edge subdivision $\graphG_4$(right). Note that the edges of $\graphG_4$ obtain orientations from the edges of $\graphG$.}
\label{fig:subdivisions}
\end{figure}

It follows immediately from~\thm{chain maps and probability} that 
\begin{proposition}\label{prop:subdivisions and expectations}
If $\graphG_n$ is the $n$-part edge subdivision of $\graphG$, and we have a measure $\mu_n$ on $\ED_n$ which is compatible with $\graphG_n$, then $\mu := (f_1^*)_\sharp \mu_n$ is compatible with $\graphG$ and for any $O(d)$ and translation-invariant measurable function $g_n$ on $\VP_n$ which can be expressed in terms of $\graphG$ as $g_n = g \circ (f_0)^*$ then
\begin{equation*}
\mean{\nu_{\graphG_n}}{g_n} = \mean{\nug}{g}
\end{equation*}
\end{proposition}
This is already useful in many cases. An easy consequence is
\begin{corollary}\label{cor:subdivisions in phantom network theory}
In James--Guth phantom network theory (edges are i.i.d. according to standard Gaussians on $\R^d$), the joint distribution of squared distances between junction vertices in $\graphG_n$ is the joint distribution of $n$ times the squared distances between vertices in $\graphG$.
\end{corollary}
We note that this same result follows from computing expected squared distance as resistance distance between junctions, and regarding the subdivided edges as composed of $n$ unit resistors in series, as in~\cite{Chen:2010da}.

Here is a second example application in James--Guth phantom network theory.
\begin{figure}
\hphantom{.}
\hfill
\includegraphics[width=2in]{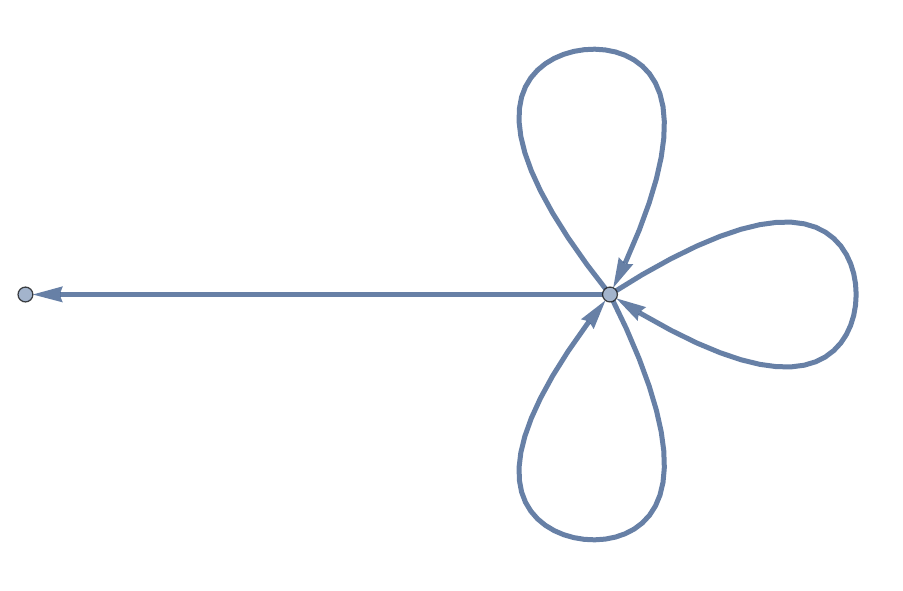}
\hfill
\includegraphics[width=2in]{theory-rewrite-sub-4.pdf}
\hfill
\hphantom{.}
\caption{In the proof of~\prop{projections}, $\graphG'$ is the loop-edge graph with 3 loops (left) and $\graphG'$ is the four-part edge subdivision of the $\alpha$-graph (right).}
\label{fig:flower}
\end{figure}

\begin{proposition}\label{prop:projections}
Suppose $\graphG$ is a connected graph with cycle rank $r$. Take any orthonormal basis $\ell_1, \dots, \ell_r$ for the loop space $\ker \bdy \subset \EC$ and any $p \in \EC$ with $\bdy p = \vertex_i - \vertex_j$. In phantom network theory (that is, when the probability measure $\mu$ on $\ED$ is a standard Gaussian) for embeddings of $\graphG$ in $\R^d$, 
\begin{equation*}
\mean{\nu_\graphG}{\norm{X(\vertex_i) - X(\vertex_j)}^2} = d \left(\ECp{p}{p} - \sum\limits_{i=1}^r \ECp{p}{\ell_i}^2 \right).
\end{equation*}
\end{proposition}

\begin{proof}
Let $\graphG'$ be the graph with two vertices $\vertex'_1$ and $\vertex'_2$, $r$ loop edges $\edge'_1, \dots, \edge'_r$ joining $\vertex_1 \rightarrow \vertex_1$, and a single edge $\edge'_{r+1}$ joining $\vertex'_1 \rightarrow \vertex'_2$.

Now define chain maps $f_0$ and $f_1$ by $f_0(\vertex'_1) = \vertex_j$ and $f_0(\vertex'_2) = \vertex_i$, while $f_1(\edge'_i) = \ell_i$ and $f_1(\edge'_{r+1}) = p$. It is easy to verify that $f_0 \bdy' = \bdy f_1$, as $\bdy f_1 (\edge'_i) = \bdy \ell_i = 0$ and $\bdy f_1(\edge'_{r+1}) = \vertex_i - \vertex_j = f_0 \bdy'(\edge'_{r+1})$. An example of this construction is shown in~\figr{flower} where $\graphG$ is a subdivision of the $\alpha$-graph (which has cycle rank $3$). 

Since $\mu$ has covariance matrix $I_d \otimes I_\edgesE$ on $\ED$, the pushforward $(f_1^*)_\sharp \mu$ has covariance matrix
\begin{equation*}
(f_1^*) (f_1^*)^T = (I_d \otimes f_1^T)(I_d \otimes f_1) = (I_d \otimes f_1^T f_1)
\end{equation*}
It follows from our definition of $f_1$ that $f_1^T f_1$ is a $2 \times 2$ block matrix with
\begin{equation*}
(f_1^T f_1)_{11} = I_{r}, \quad
(f_1^T f_1)_{12} = ( \ECp{\ell_1}{p} \cdots \ECp{\ell_r}{p} ), \quad
(f_1^T f_1)_{21} = (f_1^T f_1)_{12}^T, \quad
(f_1^T f_1)_{22} = \ECp{p}{p}. 
\end{equation*}
The $d \times d$ covariance matrix of $((f_1^*)_\sharp \mu)_\graphG$ is the conditional variance of $W(\edge_{r+1})$ conditioned on $W(\edge_1), \dotsc, W(\edge_r) = 0$.

So far, everything we have said is true for an arbitrary $\mu$ on $\ED$ with covariance matrix $I_{d\edgesE}$.  For an arbitrary $\mu$, we would need more information to continue, because the covariance matrix does not determine the conditional variance in general.

However, since we also know that $(f_1^*)_\sharp \mu$ is a Gaussian distribution in this special case, the conditional covariance matrix, which is the covariance matrix of $((f_1^*)_\sharp \mu)_\graphG$, can be computed by taking the Schur complement of $(f_1^T f_1)_{22}$ inside $f_1^T f_1$:
\begin{align*}
\cov{((f_1^*)_\sharp \mu)_\graphG}{-}{-} &= I_d \otimes ((f_1^T f_1)_{22} - (f_1^T f_1)_{12} (f_1^T f_1)_{11}^{-1} (f_1^T f_1)_{21}) \\
&= I_d \otimes (\ECp{p}{p} - \sum \ECp{\ell_i}{p}^2) = (\ECp{p}{p} - \sum \ECp{\ell_i}{p}^2) I_d.
\end{align*}
Now the expectation of $\norm{W(\vertex_i) - W(\vertex_j)} = \norm{W'(\vertex_2) - W'(\vertex_1)}$ is given by 
\begin{equation*}
\cov{((f_1^*)_\sharp \mu)_\graphG}{\onesDO}{\onesDO} = d(\ECp{p}{p} - \sum \ECp{\ell_i}{p}^2),
\end{equation*}
as claimed.
\end{proof}

\begin{figure}
\hphantom{.}
\hfill
\begin{overpic}[width=1.2in]{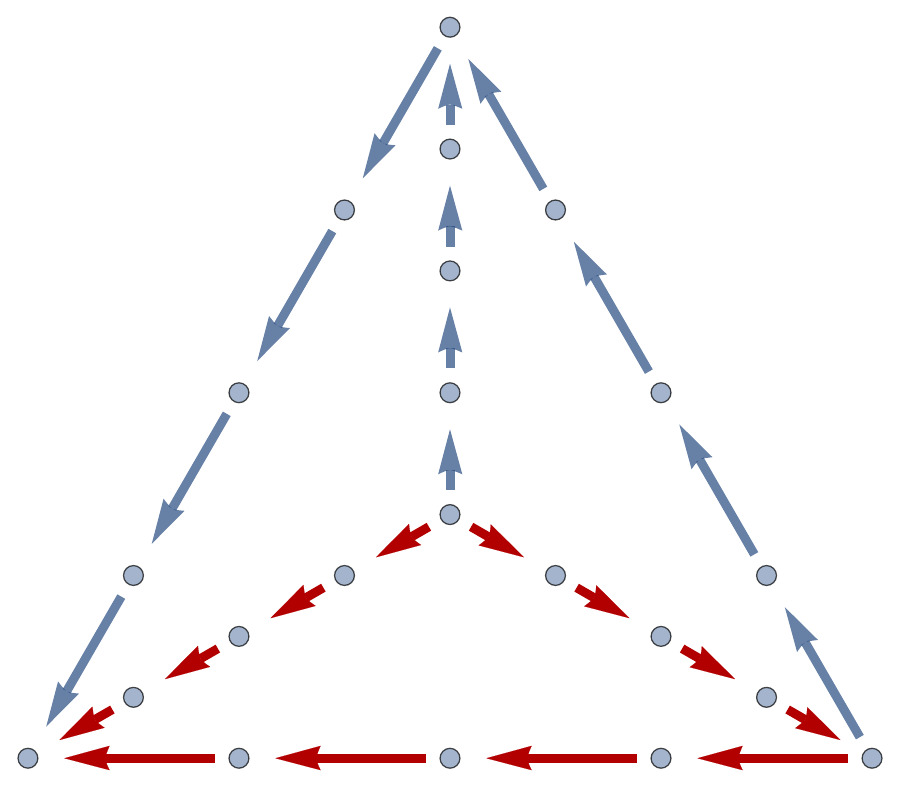}
\put(45,14){$w_1$}
\end{overpic}
\hfill
\begin{overpic}[width=1.2in]{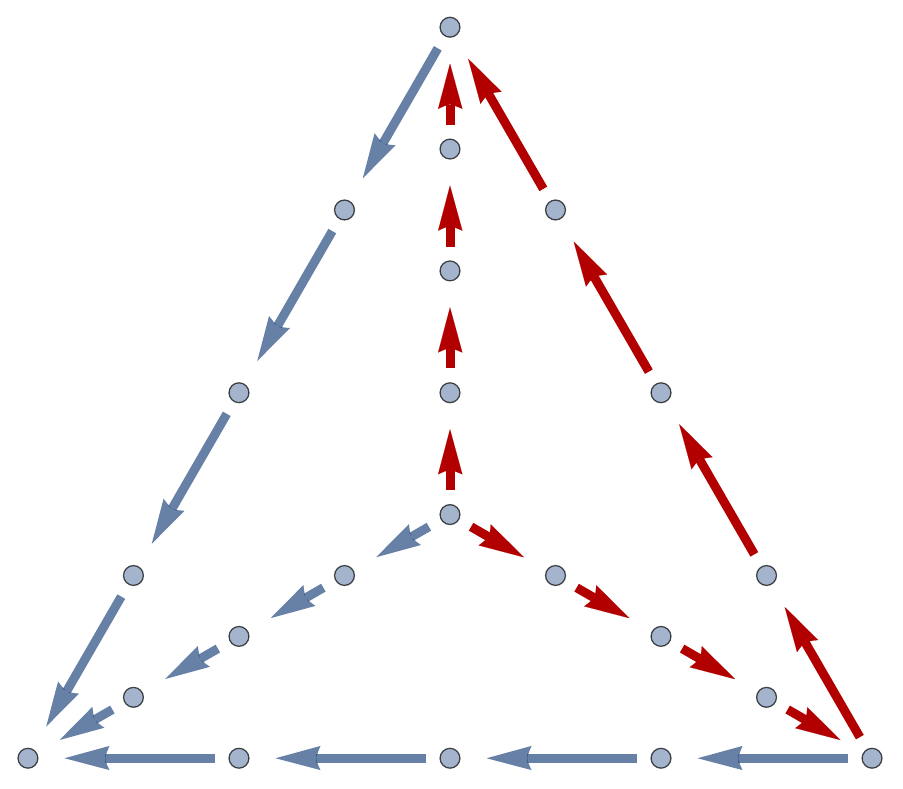}
\put(57,33){$w_2$}
\end{overpic}
\hfill
\begin{overpic}[width=1.2in]{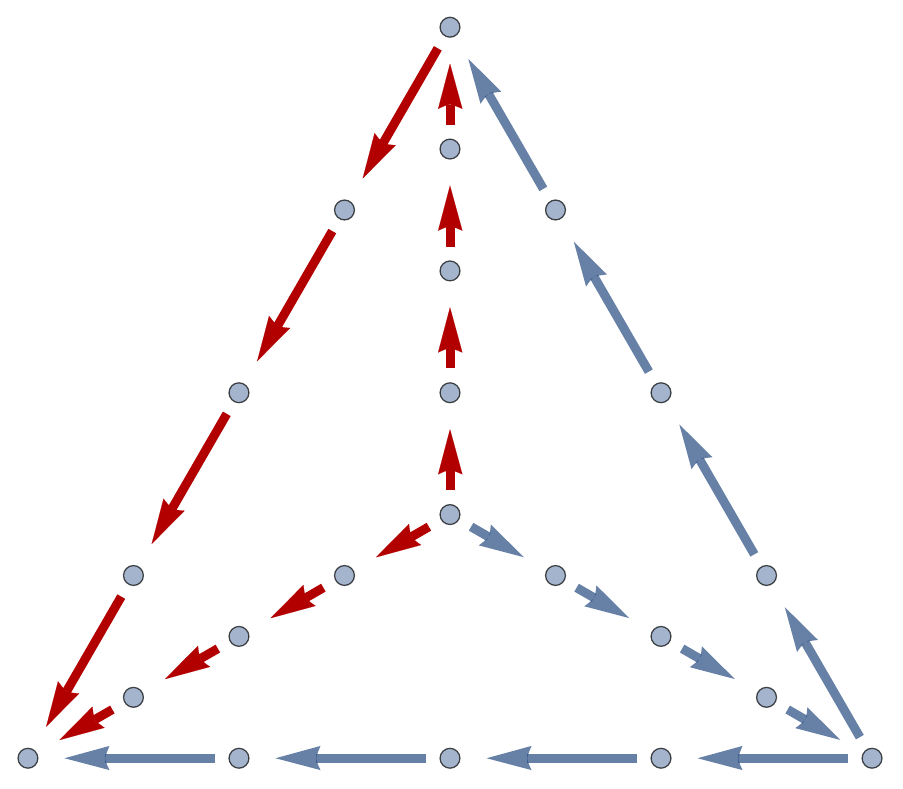}
\put(29,33){$w_3$}
\end{overpic}
\hfill
\begin{overpic}[width=1.2in]{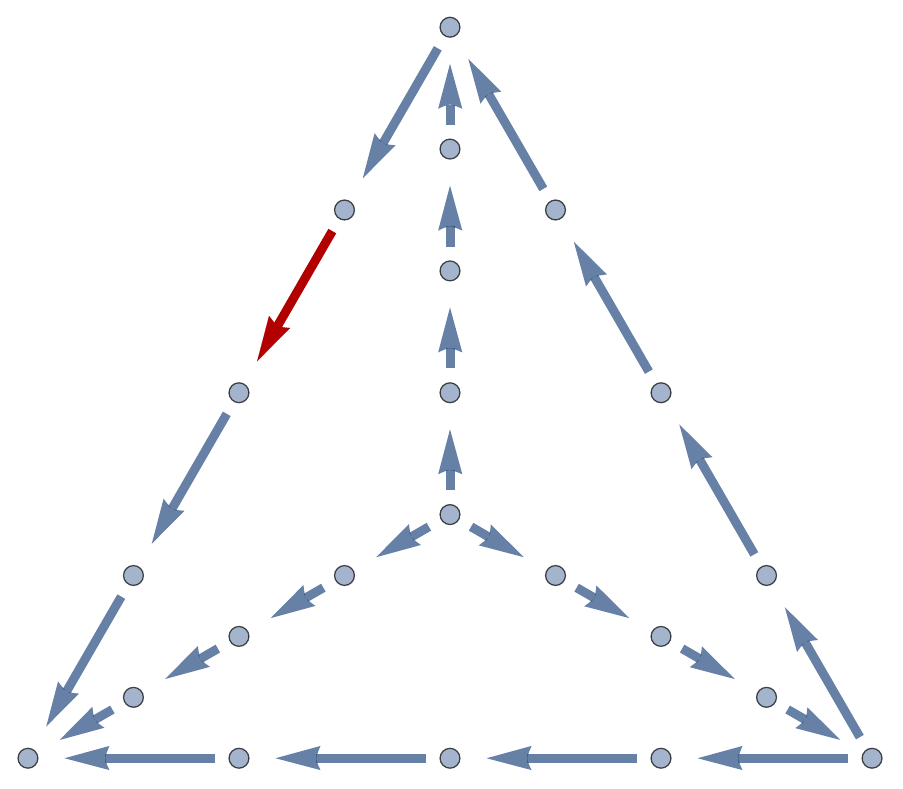}
\put(19,57){$\edge_i$}
\end{overpic}
\hfill
\hphantom{.}
\caption{On the left, we see three loops $w_1$, $w_2$ and $w_3$ which form a basis for the loop space of the subdivided $\alpha$-graph. On the right, we see a single edge $\edge_i$. Without loss of generality, we may choose corresponding $w_1$, $w_2$, $w_3$ with this relationship to an arbitrary $\edge_i$.}
\label{fig:basis}
\end{figure}

This Proposition makes it relatively easy to do particular computations in phantom network theory. For instance, we now compute the edgelength variance and junction-junction variance of the subdivided $\alpha$ graph.~\figr{basis} shows a (non-orthonormal) basis $w_1, w_2, w_3$ for the $3$-dimensional loop space of this graph, together with an edge $\edge_i$. Without loss of generality, we can assume this is the situation for any $\edge_i$. Orienting each loop counterclockwise and counting shared edges and orientations, we see that $\ECp{w_j}{w_j} = 3n$ and $\ECp{w_j}{w_k} = -n$ for all pairs of loops. We now construct an orthonormal basis $\ell_1, \ell_2, \ell_3$ by Gram-Schmidt orthogonalization:
\begin{equation*}
\ell_1 = \sqrt{\frac{1}{3n}} w_1, \quad
\ell_2 = \sqrt{\frac{1}{24n}} w_1 + \sqrt{\frac{3}{8n}} w_2, \quad
\ell_3 = \sqrt{\frac{1}{8n}} w_1 + \sqrt{\frac{1}{8n}} w_2 + \sqrt{\frac{1}{2n}} w_3.
\end{equation*}
Since $\edge_i$ is disjoint from $w_1$ and $w_2$, $\ECp{\edge_i}{w_1} = \ECp{\edge_i}{w_2} = 0$, and so $\ECp{\edge_i}{\ell_1} = \ECp{\edge_i}{\ell_2} = 0$. Since $\edge_i$ is part of $w_3$ (and agrees in orientation with $w_3$), we have $\ECp{\edge_i}{w_3} = 1$ and so $\ECp{\edge_i}{\ell_3} = \sqrt{\frac{1}{2n}}$. Thus
\begin{equation*}
\mean{\mu_\graphG}{\norm{W(\edge_i)}^2} = d\left(1 - \frac{1}{2n}\right).
\end{equation*}
To compute the expectation of the squared junction-junction distance, we replace $\edge_i$ with a sum of $n$ edges $w$ along the same subdivided edge. We get $\ECp{w}{w_1} = \ECp{w}{w_2} = \ECp{w}{\ell_1} = \ECp{w}{\ell_2} = 0$, but $\ECp{w}{w_3} = n$, and so $\ECp{w}{\ell_3} = \sqrt{\frac{n}{2}}$. The expected squared junction-junction distance in the $n$-edge subdivided $\alpha$-graph is then $dn/2$.

\prop{subdivisions and expectations} is very useful but we still have to understand $\mu_n$ well enough to establish compatibility of $\mu_n$ with $\graphG_n$ to get started. We will now show that in many cases, we can work around this limitation. 
%

\begin{proposition}\label{prop:subdivisions and compatibility}
Suppose that $\graphG_n$ is the $n$-part edge subdivision of $\graphG$. Further, suppose that we have
$n\edgesE$ independent $O(d)$-invariant probability distributions $\rho_{11}, \dots, \rho_{\edgesE n}$ on $\R^d$. Let $\rho_j$ be the joint distribution on $(\R^d)^n$ of $n$ independent vectors in $\R^d$ chosen from $\rho_{j1}, \dotsc, \rho_{jn}$. Let $\ftc \co (\R^d)^n \rightarrow \R^n$ be defined by $\ftc(x_1, \dots, x_n) = \sum x_i$.

If $\mu_n$ is the measure on $\ED_n$ obtained by choosing the $n\edgesE$ edge displacements $W(\edge_{11}), \dots, W(\edge_{\edgesE n})$ independently from $\rho_{11}, \dotsc, \rho_{\edgesE n}$, then the pushforward $f_1^* \mu_n = \mu$ is obtained by choosing the $\edgesE$ edge displacements $W(\edge_1), \dots, W(\edge_n)$ independently from $\ftc_\sharp \rho_1, \dots, \ftc_\sharp \rho_\edgesE$.

If $\mu$ is compatible with $\graphG$ and each $\rho_j$ has a decomposition with respect to $\ftc$ given by a family of measures $\rho_j^{W}$ on $(\R^d)^n$ and the pushforward $\ftc_\sharp \rho_j$, then $\mu_n$ is compatible with $\graphG_n$ and~\prop{subdivisions and expectations} holds.
\end{proposition}

\begin{proof}
We are going to construct the conditional probabilities $\mugwn$ by constructing a decomposition of $\mu_n$ with respect to $\proj_{\ID_n}$ and the measure $(\proj_{\ID_n})_\sharp \mu_n$, keeping in mind that $W_n$ is a member of $\ED_n = ((\R^d)^n)^\edgesE$. We do this in several stages. We know that we have maps 
\begin{equation*}
\ED_n \xrightarrow{f_1^*} \ED \xrightarrow{\proj_{\ID}} \ID 
\end{equation*}
We first note that $\mu_n$ is the joint distribution of independent vectors in $(\R^d)^n$ chosen from the distributions $\rho_1, \dots, \rho_e$. Now as a $d \times dn$ matrix, $\ftc = I_d \otimes \onesON$. Further, we can compute
\begin{equation} \label{eq:kronecker subdivision}
f_1^* = I_d \otimes f_1^T = I_d \otimes I_e \otimes \onesON = I_e \otimes I_d \otimes \onesON = I_e \otimes \ftc.
\end{equation}
Therefore, the pushforward $\mu = (f_1^*) \mu_n$ is the joint distribution of $\ftc_\sharp \rho_1, \dots, \ftc_\sharp \rho_\edgesE$ on $\ED = (\R^d)^n$. We have assumed that the $\rho_i$ have decompositions with respect to $\ftc$, so we may construct a family of $(\mu_n)^Z$ decomposing $\mu_n$ with respect to $f_1^*$ by defining $(\mu_n)^Z$ as the joint distribution of the decomposing distributions $\rho_1^{Z(\edge_1)}, \dots, \rho_{\edgesE}^{Z(\edge_\edgesE)}$. Further, we have assumed that $\mu$ is compatible with $\graphG$, so there are $\mugw$ decomposing $\mu$ with respect to $\proj_{\ID}$. 

Suppose we have some $g_n \in \mathcal{K}(\ED_n)$. We can define a new function $g$ by taking $g(Z) = (\mu_n)^Z(g_n)$. Since the $\mu_n^Z$ are weak$^*$-continuous in $Z$ as measures on $\ED$, their values on the fixed function $g_n$ are also a continuous function of $Z$. Further, since $g_n$ has compact support on $\ED_n$, the new function $g$ has compact support on $\ED$, and $g \in \mathcal{K}(\ED)$.

We can then define a measure $(\mu_n)^W$ on $\ED_n$ by $(\mu_n)^W(g_n) = \mugw(g)$ for each $W$ in $U$ where $\mugw$ is defined. We claim that these $(\mu_n)^W$ and the measure $(\proj_{\ID} \circ f_1^*)_\# \mu_n$ are a decomposition of $\mu_n$ with respect to $\proj_{\ID} \circ f_1^*$  

Continuity of $(\mu_n)^W$ in $W$ follows from continuity of $\mugw$ in $W$. To show that $(\mu_n)^{W}$ is concentrated on $(\proj_{\ID} f_1^*)^{-1}(W)$, we argue as follows. Suppose $A$ is an open set in $\ED_n$ which is disjoint from $(\proj_{\ID} f_1^*)^{-1}(W)$ and $\chi_A(Z_n)$ is its characteristic function. The corresponding function $g(W) = (\mu_n)^W(\chi_A)$ is supported on $(f_1^*)(A)$, but by hypothesis, $f_1^*(A)$ is disjoint from $\proj_{\ID}^{-1}(W)$. Since $\mugw$ is concentrated on $\proj_{\ID}^{-1} W$, this means that $\mugw(g) = 0$. 

We last have to check the averaging property. This is a computation:
\begin{equation*}
\int (\mu_n)^W(g_n) (\proj_{\ID} \circ f_1^*)_\# \mu_n(dW)
= \int \mugw(g) (\proj_{\ID})_\sharp \mu(dW)
= \mu(g)
= (f_1^*)_\# \mu_n(g)
= \mu_n(g_n).
\end{equation*} 

Now it is clear from the definition of the canonical chain maps that they have no kernel. Therefore they are injective. One can give a sophisticated proof that $\xi(\graphG) = \xi(\graphG_n)$ because the two spaces are homotopy equivalent and $\xi$ is the first Betti number. However, it is easier to compute 
\begin{equation*}
\xi(\graphG) = \edgesE - \verticesV + 1 = n \edgesE - (\verticesV + (n-1) \edgesE) + 1 = \xi(\graphG_n).
\end{equation*}
Therefore the hypotheses of~\prop{chain map pushforwards} hold. We've already proved in~\clm{id isomorphism} of the proof of that Proposition that $\proj_{\ID} f_1^*$ is a linear isomorphism from $\ID_n$ to $\ID$. Therefore, there is an inverse map $(\proj_{\ID} f_1^*)^{-1} \co \ID \rightarrow \ID_n$. Further, we saw in~\clm{proj f1 proj is proj f1} that $\proj_{\ID} f_1^* = \proj_{\ID} f_1^* \proj_{\ID_n}$. Pushing our measure $(\proj_{\ID} f_1^*)_\# \mu_n$ on $\ID$ forward by $(\proj_{\ID} f_1^*)^{-1}$ to $\ID_n$, we see that
\begin{equation*}
(\proj_{\ID} f_1^*)^{-1}_\sharp (\proj_{\ID} f_1^*)_\sharp \mu_n = 
(\proj_{\ID} f_1^*)^{-1}_\sharp (\proj_{\ID} f_1^* \proj_{\ID_n})_\sharp \mu_n =
(\proj_{\ID_n})_\sharp \mu_n.
\end{equation*}
Thus we can define measures $\mugwn := (\mu_n)^{(\proj_{\ID} f_1^*)^{-1}(W_n)}$ which decompose $\mu_n$ with respect to the map $\proj_{\ID_n}$ and the measure $(\proj_{\ID_n})_\sharp \mu_n$. By~\thm{decomposition and existence}, this shows that $\mu_n$ is compatible with $\graphG_n$.
\end{proof}
We note that this proposition also covers generalized subdivisions of $\graphG$ where the number of subdivisions of each edge varies between the edges of $\graphG$; this can be proved by choosing $n$ to be largest number of subdivisions and setting unused $\rho_{ij}$ to $\delta(0)$ so that some ``edges'' are forced to have length~$0$. Alternatively, one can repeat the proof above-- the only difficulties in writing the analogue of~\eqn{kronecker subdivision} are notational.

In particular, let's consider a generalization of the freely-jointed chain.
\begin{definition}
If $\graphG_n$ is a $n$-part edge subdivision of any graph $\graphG$ with $n \geq 3$, and $\mu_n$ is the joint distribution of independent edge displacements chosen from the area measure on $S^2 \subset \R^3$, we will call $\graphG_n$, $\mu_n$ a~\emph{freely jointed network} with~\emph{structure graph} $\graphG$.
\end{definition}

\begin{proposition}\label{prop:freely jointed network}
The measure $\mu_n$ in the freely jointed network $\graphG_n$ is compatible with $\graphG_n$. The corresponding measure $\mu$ on the structure graph $\graphG$ independently samples edge displacements from  
\begin{equation} \label{eq:freely jointed end-to-end}
\rho(x) = \left(\frac{1}{2\pi^2 \ell} \int_0^\infty y \sin \ell y \sinc^n y \,d{y} \right) \lambda^{3}(dW)
\end{equation}
where $\ell = \norm{x}$. Further, any function $g_n \co \VP_n \rightarrow \R$ which can be expressed in terms of $\graphG$ as $g_n = g \circ f_0^*$ has 
\begin{equation*}
\mean{\nu_{\graphG_n}}{g_n} = \mean{\nu_{\graphG}}{g}.
\end{equation*}
\end{proposition}

\begin{proof}
This is a combination of our existing results. $\mu$ is compatible with $\graphG$ by~\prop{disintegration with density} because it has a continuous density given by the product of the density of $\rho(x)$ which is 
positive in a neighborhood of the origin. By~\prop{conditional probability for arms}, $\rho$ has a decomposition with respect to $\ftc$ and $\ftc_\sharp \rho$, so we can apply~\prop{subdivisions and compatibility} to show that $\mu_n$ is compatible with $\graphG_n$. Now we can apply~\thm{chain maps and probability} to complete the result.
\end{proof}

\prop{freely jointed network} makes the computation of many expectations quite feasible for arbitrary freely jointed networks. We now describe an example numerical computation using~\prop{freely jointed network}. Suppose that $\graphG$ is the $\alpha$-graph (a.k.a., the complete graph $K_4$) and we consider the freely jointed network with graph $\graphG_n$ in $\R^3$. The graph $\graphG$ has $\verticesV = 4$ and $\edgesE = 6$, so the cycle rank $\xi(\graphG) = 3$. Therefore $\ID$ is $d\xi(\graphG) = 9$ dimensional, $\im \bdy^*$ is $d(\verticesV-1) = 9$ dimensional, and $\ED$ is $d \edgesE = 18$ dimensional. 

We parametrized centered configurations of four vertices in $\R^3$ ($\im \bdy^{*+} \subset \VP$) by $\R^9 = (\R^3)^3$ using
\begin{equation*}
(\vec{x}_1,\vec{x}_2,\vec{x}_3) \mapsto \frac{1}{4}(\vec{x}_1 + \vec{x}_2 + \vec{x}_3, \vec{x}_1 - \vec{x}_2 - \vec{x}_3, -\vec{x}_1 + \vec{x}_2 - \vec{x}_3, -\vec{x}_1 - \vec{x}_2 + \vec{x}_3), 
\end{equation*}
and composed with $\bdy^*$ to parametrize $\im \bdy^* \subset \ED$ by 
\begin{equation*}
(\vec{x}_1,\vec{x}_2,\vec{x}_3) \mapsto \frac{1}{2} \left(\vec{x}_1 - \vec{x}_2, \vec{x}_1 + \vec{x}_2, \vec{x}_1 + \vec{x}_3, \vec{x}_1 - \vec{x}_3, \vec{x}_2 + \vec{x}_3, \vec{x}_2 - \vec{x}_3 \right).
\end{equation*}
Now the (unnormalized) probability density for a given configuration is given by the product of $\rho$ from \eqn{freely jointed end-to-end} evaluated on the six edge displacements above. We found the partition function $m_0$ for $n$ between $3$ and $10$ by performing a $6$-dimensional numerical integral\footnote{Reduced from a 9-dimensional integral using the $O(3)$-symmetry.} for each $n$. We emphasize that although the dimension of $\ED(\graphG_n)$ rises with $n$, the dimension of $\ED(\graphG)$ does not, so these integrals were all of comparable difficulty. Similarly, we were able to (numerically) integrate the squared length $\norm{W(\edge_1)}$ over this space to compute the expectation of squared junction-junction distance. We compared these results to the averages over 10,000 samples from the Markov chain method of Deguchi and Uehara~\cite{Uehara:2018bb} for freely jointed networks with maximum vertex degree 3, where we made 1,000 random moves between samples. The results are shown in~\figr{numerical integration versus markov}. They are quite close, supporting the conjecture that the Markov chain is converging to the correct measure. 
\begin{figure}[t]
\hphantom{.}
\hfill
\includegraphics[height=2in]{theory-rewrite-sub-4.pdf}
\hfill
\begin{overpic}[height=2in]{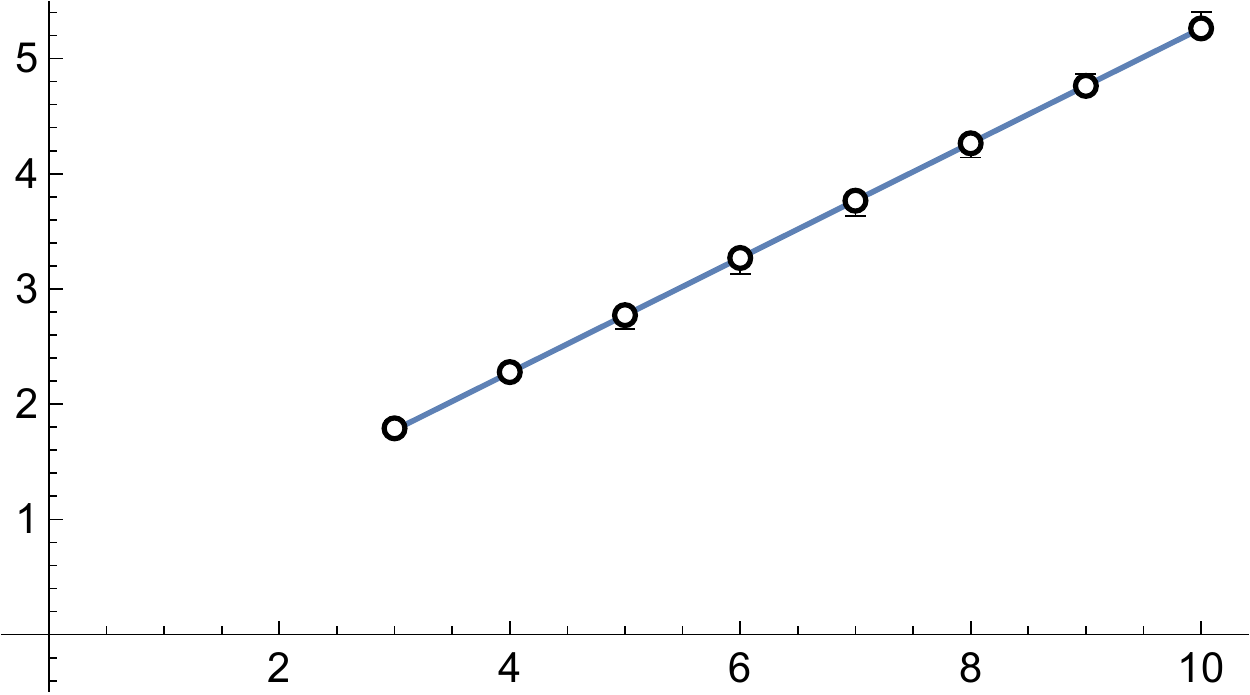}
\put(15,-3){Number of subdivisions $n$ of each edge of $\alpha$-graph}
\put(-22,60){\begin{minipage}{2in}Expectation of squared\\ junction-junction distance\end{minipage}}
\end{overpic}
\hfill
\hphantom{.}
\vspace{0.1in}
\caption{The right-hand graph shows the expectation of the squared distance between junctions in freely jointed networks obtained by subdividing the $\alpha$-graph (as shown at left). The circles are results obtained by $6$-dimensional numerical integration (following the discussion after~\prop{freely jointed network}) while the fences are 95\% confidence intervals for Monte Carlo integration using the method of~\cite{Uehara:2018bb}. The linear fit is to a line of slope $0.497981 \simeq 0.5$.}
\label{fig:numerical integration versus markov}
\end{figure}

\section{Conclusion}
We have now given a theory of random embeddings of graphs with respect to a very general class of probability distributions on the edges. From a mathematical point of view, it would be interesting to see how much further these results can be pushed. We established our theory for freely jointed networks by carefully proving the existence of conditional distributions for the freely jointed arm. This is not yet conclusive: for instance, what if we had fixed bond~\emph{angles} instead of lengths?

An alternative (and more standard) approach to the theory above would be to build conditional probabilities via~\emph{disintegrations} (cf.~\cite{Chang1997}) rather than decompositions. This allows one to establish the existence of a conditional $\mugw$ for~\emph{almost every} $W \in U \subset \ID$ in our theorems above. The only hypothesis needed for this approach is that the pushforward measure $\mu_{\ID}$ has a density with respect to Lebesgue measure on~$\ID$. We have not followed this path above because our primary interest is in cases where one can build a single well-defined probability distribution $\mug$. 

It has not escaped our attention that the explicit construction of $(\mu_n)^0$ in~\prop{subdivisions and compatibility} suggests various explicit sampling algorithms, particularly for freely jointed networks. We will develop these in a future publication. Last, we note that when one is considering problems with self-avoidance or steric constraints, the relevant graph is clearly the complete graph, where the bonds and the repulsive forces are distinguished by different probability distributions on different edges. In this case, there are various useful simplifications to be made to the theory above. We hope to say more about this in the future.

\section*{Acknowledgments}
The authors would like to acknowledge many friends and colleagues whose helpful discussions and generous explanations shaped this work. In particular we would like to acknowledge Yasuyuki~Tezuka and Satoshi Honda for helpful discussions of topological polymer chemistry and thank Fan Chung for introducing us to spectral graph theory. This paper stemmed from a long series of discussions which started at conferences at Ochanomizu University and the Tokyo Institute of Technology. Cantarella and Shonkwiler are grateful to the organizers and the Japan Science and Technology Agency for making these possible. In addition, we are grateful for the support of the Simons Foundation (\#524120 to Cantarella, \#354225 and \#709150 to Shonkwiler), the Japan Science and Technology Agency (CREST Grant Number JPMJCR19T4) and the Japan Society for the Promotion of Science (KAKENHI Grant Number JP17H06463).

\bibliography{tcrwpapers-special,tcrwpapers,tcrw-cites,TCRW-citations}

\appendix

\section{Linear algebra background}
\label{sec:background}

In what follows we will typically use the same symbol---e.g., $F$---to refer both to a linear map and to the matrix for that linear map with respect to given bases on domain and range. We will do this without comment unless the distinction is important or the chosen bases are not clear from context.

\begin{definition}\label{def:hom}
Given vector spaces $V$ and $W$, the vector space $\Hom(V,W)$ is the space of linear maps $A \co V \rightarrow W$. If $V$ is $n$-dimensional and $W$ is $m$-dimensional, then $\Hom(V,W)$ is $mn$-dimensional. 
If we choose bases $v_1, \dots , v_n$ for $V$ and $w_1, \dots, w_m$ for $W$, there is a natural basis $E_{11}, \dotsc, E_{mn}$ of linear maps defined by $E_{ij}(x_1 v_1 + \cdots + x_n v_n) = x_j w_i$. We always assume that the $E_{ij}$ are presented in lexicographic order on $i,j$: that is, as $E_{11}, \dotsc, E_{1n}, E_{21}, \dotsc, E_{mn}$.
\end{definition}

We note that (given bases for $V$ and $W$), we can also think of $\Hom(V,W)$ as the space of $m \times n$ matrices; thought of as a matrix, $\mat{E_{ij}}$ is the $m \times n$ matrix with $1$ in position $(i,j)$ and $0$s everywhere else. 

If $U$, $V$, and $W$ are vector spaces, any linear map $F\co U \to V$ induces a linear map $F^*\co \Hom(V,W) \to \Hom(U,W)$, where $A \mapsto F^* A$ is defined by
\begin{equation}\label{eq:induced hom map}
	(F^*A)(u) := A(F u)
\end{equation}
for any $u \in U$. Any such map is linearly equivariant:

\begin{lemma}\label{lem:linear equivariance}
	If $F\co U \to V$ is linear, the induced map $F^*\co \Hom(V,W) \to \Hom(U,W)$ is $GL(W)$-equivariant.
\end{lemma}
\begin{proof}
	By definition, for any $X \in \Hom(V,W)$, $H \in GL(W)$, and $u \in U$,
	\[
		(F^* H X)u = HX(Fu) = H(XF)u = (H F^* X)u,
	\]
	so $F^*HX = HF^*X$.
\end{proof}

Since $F^* A \in \Hom(U,W)$, it can be represented by an $m \times k$ matrix with respect to bases for $U$ and $W$, namely $AF$, as we see in \eqref{eq:induced hom map}. On the other hand, $\Hom(V,W)$ is $mn$-dimensional and $\Hom(U,W)$ is $mk$-dimensional, with bases given by Definition~\ref{def:hom}. So the matrix for $F^*$ with respect to these bases should be an $mk \times mn$ matrix
\begin{equation}
	\mat{F^*} = \mat{I_m} \otimes \mat{F}^T.
\label{eq:Fstar formula}
\end{equation}

Here, the Kronecker product matrix $\mat{I_m} \otimes \mat{F}^T$ is the $mk \times mn$ block matrix consisting of blocks of size $k \times n$, where the diagonal blocks are copies of the $k \times n$ matrix $\mat{F}^T$ and the off-diagonal blocks consist of zeros.\footnote{This choice of notation corresponds to the fact that $\Hom(V,W) \simeq W \otimes \Hom(V,\R)$.} More generally:

\begin{definition}\label{def:kronecker product}
	If $A$ is an $m \times n$ matrix with entries $a_{ij}$ and $B$ is a $k \times \ell$ matrix, then the \emph{Kronecker product} $A \otimes B$ is the $km \times \ell n$ matrix
	\[
		A \otimes B = \begin{bmatrix} a_{11} B & \cdots & a_{1n} B \\ 
									   \vdots & \ddots & \vdots \\
									   a_{m1} B & \cdots & a_{mn} B \end{bmatrix}.	
	\]
\end{definition}

At various times in what follows we will need to compute the usual matrix product of matrices expressed as Kronecker products. These \emph{mixed products} can be computed as follows:

\begin{lemma}\label{lem:mixed product}
	If $A$, $B$, $C$, and $D$ are matrices of appropriate sizes so that the products $AC$ and $BD$ make sense, then the mixed product
	\[
		(A \otimes B)(C \otimes D) = (AC) \otimes (BD).
	\]
\end{lemma}

Coming back to induced maps on Hom spaces, it's a general fact that we have $(FG)^* = G^* F^*$ whenever $F$ and $G$ are composable linear maps. The special case where $W = \R$ is important:

\begin{definition}\label{def:the transpose operator}
Given bases $v_1, \dots , v_n$ for $V$ and $u_1, \dots, u_k$ for $U$, then $\Hom(V,\R) \simeq V$ and $\Hom(U,\R) \simeq U$ by isomorphisms $v_1, \dots, v_n \rightarrow E_{11}, \dots, E_{1n}$ and $u_1, \dots, u_k \rightarrow E_{11}, \dots, E_{1k}$. (We call both of these isomorphisms $\star$.) Then any linear map $F \co V \rightarrow U$ has a corresponding map $F^T = \star F^* \star \co U \rightarrow V$. Of course, the matrix  of $F^T$ is simply the transpose of the matrix for $F$.
\end{definition}

For a linear map $F$, it will be helpful to identify the kernel and image of $F^*$ in terms of subspaces associated to $F$, which we will be able to do using annihilators: 
\begin{definition} \label{def:annihilator}
If $S \subset V$ is a subspace, then the~\emph{annihilator} $S^0 \subset \Hom(V,W)$ is the set of $X \in \Hom(V,W)$ with $S \subset \ker X$. 
\end{definition}

\begin{proposition} \label{prop:annihilator props}
If we have a map $F \co U \rightarrow V$, and we take $F^* \co \Hom(V,W) \rightarrow \Hom(U,W)$ then
\begin{equation*}
\im F^* = (\ker F)^0 \quad\text{and}\quad \ker F^* = (\im F)^0.
\end{equation*}
\end{proposition}

When $F$ is injective, then $\ker F = \{0\}$, which is annihilated by everything; when $F$ is surjective, then $\im F$ is all of $V$, which is only annihilated by the zero element. In other words:

\begin{corollary}\label{cor:injective dual to surjective}
	Suppose $U$, $V$, and $W$ are vector spaces and $F\co U \to V$. If $F$ is injective then $F^*\co \Hom(V,W) \to \Hom(U,W)$ is surjective, and if $F$ is surjective then $F^*$ is injective.
\end{corollary}

Now we add more structure to our vector spaces. Specifically, suppose that $U$ and $V$ are inner product spaces with inner products $\Up{-}{-}$ and $\Vp{-}{-}$. Then a linear map $F \co U \to V$ induces a map going the other way:

\begin{definition}
If $F \co U \to V$ is a linear map between inner product spaces, then the \emph{adjoint map} $F^\dag \co V \rightarrow U$ is the unique linear map so that, for any $u \in U$ and $v \in V$, 
\[
	\Vp{Fu}{v} = \Up{u}{F^\dag v}. 
\]
If $F \co V \to V$ and $F = F^\dag$, we say that $F$ is \emph{self-adjoint}.
\end{definition}

When the inner product is the standard one, $F^T = F^\dag$. Like transpose, $(FG)^\dag = G^\dag F^\dag$ for compositions of maps between inner product spaces. The following lemma is a familiar fact in that setting, but must be rewritten as below when using a different inner product.
\begin{lemma}\label{lem:orthogonal decomposition}
If $F \co U \rightarrow V$ is a map between inner product spaces, then $V = \im F \oplus \ker F^\dag$ and $U = \ker F \oplus \im F^\dag$ are orthogonal decompositions.
\end{lemma} 

Just as bases on $V$ and $W$ induce a basis for $\Hom(V,W)$, inner products on $V$ and $W$ induce an inner product on $\Hom(V,W)$.

\begin{definition}\label{def:induced inner product}
	If $\Vprod$ and $\Wprod$ are inner product spaces, then the \emph{Frobenius inner product} on $\Hom(V,W)$ is given by
	\[
		\dotp{A}{B}_{\Fr} := \tr A^\dag B.
	\]
\end{definition}

In practice, we often identify inner products with self-adjoint operators as follows: if $V$ is $n$-dimensional, then it is abstractly isomorphic with $\R^n$. A choice of basis $v_1, \dots , v_n$ for $V$ determines a specific isomorphism $V \to \R^n$ which sends each $v_i$ to the $i$th standard basis vector. Under this identification, the standard dot product on $\R^n$ defines an inner product $\dotp{-}{-}_{\std}$ on $V$ by 
\[
	\dotp{x_1 v_1 + \dots + x_n v_n}{y_1 v_1 + \dots + y_n v_n}_{\std} = \sum_{i=1}^n x_i y_i.
\]

It is a standard fact that every inner product on $V$ can be written in terms of this inner product:

\begin{proposition}\label{prop:nonstandard inner product}
	Suppose $V$ is a vector space with basis $v_1, \dots , v_n$ and corresponding standard inner product $\dotp{-}{-}_{\std}$. If $H\co V \to V$ is $\dotp{-}{-}_{\std}$-self-adjoint\footnote{Equivalently, the matrix for $H$ with respect to the basis $v_1, \dots , v_n$ is symmetric.} and positive-definite, then
	\[
		\dotp{x}{y}_H := \dotp{x}{Hy}_{\std}
	\]
	is an inner product on $V$.
	
	Conversely, every inner product on $V$ is equal to $\dotp{x}{y}_H$ for some $\dotp{-}{-}_{\std}$-self-adjoint, positive-definite $H: V \to V$.
\end{proposition}

In this notation, $\dotp{-}{-}_{\std} = \dotp{-}{-}_I$, where $I$ is the identity map on $V$.

We will repeatedly use the following lemma, which gives the matrix of the adjoint of a linear map:

\begin{lemma}\label{lem:adjoint formula}
	Suppose $V$ and $W$ are finite-dimensional vector spaces with bases $v_1, \dots , v_n$ and $w_1, \dots, w_m$, respectively, and corresponding standard inner products $\dotp{-}{-}_{V,\std}$ and $\dotp{-}{-}_{W,\std}$. Suppose $F \co V \to V$ and $H \co W \to W$ are positive-definite and self-adjoint with respect to the standard inner products. If $A \co V \to W$ is linear, then the adjoint of $A$ with respect to the $\dotp{-}{-}_{V,F}$ and $\dotp{-}{-}_{W,H}$ inner products is represented by the matrix
	\[
		\mat{A^\dag} = \mat{F}^{-1} \mat{A}^T \mat{H}.
	\]
\end{lemma}

\begin{proof}
	By definition of the adjoint,
		\begin{equation}\label{eq:theta adjoint}
			\dotp{A v}{w}_{W,H} = \dotp{v}{A^\dag w}_{V,F}
		\end{equation}
		for any $v \in V$ and $w \in W$. In turn, the left hand side is equal to
		\[
			\dotp{A v}{H w}_{W,\std} = \dotp{v}{\mat{A}^T \mat{H} w}_{V,\std}
		\]
		and the right hand side is equal to
		\[
			\dotp{v}{F A^\dag w}_{V,\std}.
		\]
		Equating these two and recalling that equation~\eqref{eq:theta adjoint} must hold for all $v$ and $w$, we conclude that $\mat{A}^T \mat{H}  = \mat{F} \mat{A^\dag}$ or, since $F$ is invertible,
		\[
			\mat{A^\dag} = \mat{F}^{-1} \mat{A}^T \mat{H}.
		\]
\end{proof}

\begin{proposition}\label{prop:hom inner product}
	Under the same hypotheses as in Lemma~\ref{lem:adjoint formula}, the Frobenius inner product on $\Hom(V,W)$ induced by $\dotp{-}{-}_{V,F}$ and $\dotp{-}{-}_{W,H}$ is
	\[
		\dotp{-}{-}_{\Fr} = \dotp{-}{-}_{H \otimes F^{-1}}
	\]
\end{proposition}

\begin{proof}
	Suppose $A,B \in \Hom(V,W)$. Then, by definition, 
	\[
		\dotp{A}{B}_{\Fr} = \tr A^\dag B,
	\]
	where $A^\dag$ is the adjoint of $A$ with respect to the $\dotp{-}{-}_{V,F}$ and $\dotp{-}{-}_{W,H}$ inner products. Hence, Lemma~\ref{lem:adjoint formula} implies that
\begin{equation}
		\dotp{A}{B}_{\Fr} = \tr \mat{A^\dag}\mat{B} = \tr \mat{F}^{-1} \mat{A}^T \mat{H}\mat{B} = \tr \mat{A}^T \mat{H}\mat{B} \mat{F}^{-1}  = \dotp{A}{B}_{H \otimes F^{-1}}
\label{eq:Frobenius product formulas}
\end{equation}
	by the cyclic invariance of trace.
\end{proof}

In general, if a linear map $F \co U \to V$ is not bijective, it cannot be invertible, but the following defines a map $V \to U$ which is, in a sense, as close as possible to being an inverse for $F$:

\begin{definition}\label{def:pseudoinverse}
If $F \co U \to V$ is a linear map between inner product spaces, the \emph{Moore--Penrose pseudoinverse} $F^+ \co V \to U$ is the unique linear map satisfying:
\begin{equation}
F F^+ F = F, \quad F^+ F F^+ = F^+, \quad FF^+ \text{ and } F^+F \text{ are self-adjoint.}
\label{eq:Moore-Penrose conditions}
\end{equation}
\end{definition}

If $D \co U \to V$ is represented in terms of bases for $U$ and $V$ by a diagonal $k \times n$ matrix with diagonal entries $d_1, \dots, d_{\min(k,n)}$, then $D^+ \co V \to U$ is written in terms of the same bases as an $n \times k$ diagonal matrix with diagonal entries 
\begin{equation*}
d^+_i = 
\begin{cases}
\frac{1}{d_i}, & \text{if $d_i \neq 0$}, \\
0, & \text{if $d_i = 0$}.
\end{cases}
\end{equation*}
More generally, if $F \co U \to V$ has the singular value decomposition $\mat{F} = \mat{\Phi}\mat{\Sigma}\mat{\Psi^\dag}$,\footnote{SVDs are usually written in the form $\mat{U} \mat{\Sigma} \mat{V}^T$, but we are already using $U$ and $V$ as the names of our vector spaces. Further, we have to recall that an orthogonal matrix on an inner product space is one whose inverse is its adjoint, which, as we note in Lemma~\ref{lem:adjoint formula}, is not necessarily its transpose.} then the pseudoinverse ${F^+ \co V \to U}$ is represented by the matrix $\mat{F^+} = \mat{\Psi} \mat{\Sigma^+} \mat{\Phi^\dag}$.
				
The Moore--Penrose pseudoinverse has many useful properties. 
\begin{proposition}\label{prop:pseudoinverse properties}
If $Fu = v$, then $\widehat{u} = F^+ v$ minimizes $\Vnorm{Fu -v }^2 = \Vp{Fu-v}{Fu-v}$. Further, $F F^+$ is orthogonal projection (with respect to $\Vp{-}{-}$) onto $\im F$, while $F^+ F$ is orthogonal projection (with respect to $\Up{-}{-}$) onto $\im F^\dag$. In addition, $\ker F^+ = \ker F^\dag$ and $\im F^+ = \im F^\dag$. Further, $(F^\dag)^+ = (F^+)^\dag$ so we can write $F^{\dag +}$ without ambiguity.
\end{proposition}

In general, we will denote orthogonal projection onto a subspace $S$ of an inner product space $\Vprod$ by $\proj_S$; it is understood that the projection is orthogonal with respect to the inner product $\Vp{-}{-}$.

In the description of the properties of the Moore--Penrose pseudoinverse above, we noted that $(F^+)^\dag = (F^\dag)^+$. In fact, this property generalizes even further, motivating the definition of the induced inner product on $\Hom$ spaces.

\begin{proposition}\label{prop:starplus is plusstar}
Suppose we have inner product spaces $\Uprod$, $\Vprod$, and $\Wprod$, together with $\Hom(U,W)$ and $\Hom(V,W)$ with their induced (Frobenius) inner products, and a linear map $F \co U \rightarrow V$. Then $(F^*)^+ = (F^+)^*$. Both are maps $\Hom(U,W) \rightarrow \Hom(V,W)$.
\end{proposition}

\begin{proof}
The proof is an exercise in checking that $(F^+)^*$ satisfies the Moore--Penrose conditions for a pseudoinverse of $F^* \co \Hom(V,W) \rightarrow \Hom(U,W)$.
\begin{equation*}
(F^+)^* F^* (F^+)^* = (F^+ F F^+)^* = (F^+)^* \quad\text{and}\quad
F^* (F^+)^* F^* = (F F^+ F)^* = F^*.
\end{equation*}
Referring to~\defn{induced inner product} for the Frobenius inner product, and using standard properties of $\dag$ and $*$\footnote{Here we've used several properties of $\dag$ and $\star$: $(FG)^\dag = G^\dag F^\dag$ in the second equality; $F^{*\dag} = F^{\dag*}$ and $(F^+F)^\dag = F^+ F$ in the third, and $(FG)^* = G^* F^*$ in the last one. }
\[
\Frp{F^*(F^+)^* A}{B} = \tr (F^*(F^+)^*A)^\dag B = \tr A^\dag (F^+ F)^{*\dag} B = \tr A^\dag (F^+F)^* B = \Frp{A}{F^* (F^+)^* B},
\]
so $(F^* (F^+)^*)^\dag = (F^* (F^+)^*)$. Similarly, 
\begin{multline*}
\Frp{(F^+)^* F^* A}{B} = \tr ((F^+)^* F^* A)^\dag B = \tr A^\dag ((F^+)^* F^*)^\dag B \\
=\tr A^\dag (F F^+)^{*\dag} B = \tr A^\dag (F F^+)^* B 
= \Frp{A}{(F^+)^* F^* B},
\end{multline*}
so $((F^+)^* F^*)^\dag = (F^+)^* F^*$.
\end{proof}

\end{document}